\documentclass[11pt,a4paper]{article}

\usepackage{amsmath,amssymb,amsthm}
\usepackage{amsfonts}
\usepackage{mathtools}
\usepackage{booktabs}
\usepackage{multirow}
\usepackage{enumitem}
\usepackage{array}
\usepackage[margin=2.5cm]{geometry}
\usepackage[colorlinks=true,linkcolor=blue,citecolor=red,urlcolor=blue]{hyperref}
\allowdisplaybreaks

\newtheorem{theorem}{Theorem}[section]
\newtheorem{corollary}[theorem]{Corollary}
\newtheorem{lemma}[theorem]{Lemma}
\newtheorem{proposition}[theorem]{Proposition}
\theoremstyle{definition}
\newtheorem{definition}[theorem]{Definition}
\newtheorem{remark}[theorem]{Remark}
\newtheorem{example}[theorem]{Example}

\newcommand{\Z}{\mathbb{Z}}

\newcommand{\C}{{\mathcal C}}
\newcommand{\cH}{{\mathcal H}}

\newcommand{\cA}{{\mathcal A}}

\newcommand{\cM}{{\mathcal M}}
\newcommand{\zero}{{\mathbf{0}}}
\newcommand{\one}{{\mathbf{1}}}
\newcommand{\two}{{\mathbf{2}}}
\newcommand{\four}{{\mathbf{4}}}
\newcommand{\six}{{\mathbf{6}}}
\newcommand{\uu}{\mathbf{u}}
\newcommand{\vv}{\mathbf{v}}
\newcommand{\ww}{\mathbf{w}}
\newcommand{\xx}{\mathbf{x}}

\newcommand{\zz}{\mathbf{z}}

\newcommand{\wt}{{\rm wt}}
\newcommand{\rank}{\operatorname{rank}}
\newcommand{\kernel}{\operatorname{ker}}
\newcommand{\ord}{\operatorname{o}}
\newcommand{\K}{\operatorname{K}}
\newcommand{\Span}{\operatorname{span}}

\newcommand{\sym}{\operatorname{sym}}

\begin{document}

\title{\bf Rank and classification of $\Z_2\Z_4\Z_8$-linear\\ Hadamard codes\thanks{This work has been granted by the Juan de la Cierva 2024 grant (JDC2024-053082-I), funded by Ministerio de Ciencia, Innovación y Universidades (MICIU/AEI/10.13039/501100011033) and co-funded by the European Social Fund Plus (FSE+).}}

\author{
Dipak K. Bhunia\\
\small Departament de Matem\`atiques,\\ \small Universitat Polit\`ecnica de Catalunya,\\
\small Barcelona, Spain
}
\date{}

\maketitle

\begin{abstract}
The $\Z_2\Z_4\Z_8$-additive codes are subgroups of $\Z_2^{\alpha_1} \times \Z_4^{\alpha_2} \times \Z_8^{\alpha_3}$. A $\Z_2\Z_4\Z_8$-linear  Hadamard code is a Hadamard code which is the Gray map image of a $\Z_2\Z_4\Z_8$-additive code. A recursive construction of $\Z_2\Z_4\Z_8$-additive Hadamard codes $\cH^{t_1,t_2,t_3}$ of type $(\alpha_1,\alpha_2,\alpha_3;t_1,t_2,t_3)$ with $\alpha_1\not=0$, $\alpha_2\not=0$, $\alpha_3\not=0$, $t_1\geq 1$, $t_2\geq 0$ and $t_3\geq 1$ is known, as well as the linearity and the kernel of the corresponding $\Z_2\Z_4\Z_8$-linear Hadamard codes $H^{t_1,t_2,t_3}$ of length $2^t$. The rank of these codes, however, was not known: the values used in the literature had to be computed with a computer algebra system, one type at a time.

In this paper, we determine $\rank(H^{t_1,t_2,t_3})$ for every admissible
triple $(t_1,t_2,t_3)$ and obtain an explicit closed formula.  The proof
computes the dimension of the space spanned by the binary coordinate
functions of the code, viewed as Boolean functions of the binary digits of
the message.  We then use the rank formula, together with the known dimension
of the kernel, to classify the family as far as these two invariants allow.
Inside the family, we prove that the only pairs of distinct types of the same
length sharing both invariants are
$\bigl(H^{1,2,t-6},H^{2,0,t-5}\bigr)$ for $t\geq7$ and
$\bigl(H^{2,2,t-9},H^{3,0,t-8}\bigr)$ for $t\geq10$. Consequently, writing $r$ for the rank and $k$ for the dimension of the kernel, the number of distinct pairs $(r,k)$ among the codes of length $2^t$ is $\left\lfloor(t^2+6)/12\right\rfloor$ for $3\leq t\leq6$,
$\left\lfloor(t^2+6)/12\right\rfloor-1$ for $7\leq t\leq9$, and
$\left\lfloor(t^2+6)/12\right\rfloor-2$ for $t\geq10$.  In particular, the
rank and the dimension of the kernel classify the family completely if and
only if $3\leq t\leq6$. We also determine all coincidences of these two invariants with the
$\Z_4$-linear, $\Z_2\Z_4$-linear, and $\Z_8$-linear Hadamard codes of the
same length.  There is exactly one infinite family of coincidences with the
$\Z_4$-linear and $\Z_2\Z_4$-linear families, namely
$H^{2,1,t-7}$ and $H^{5,t-9}$ for $t\geq10$, and exactly five families of
coincidences with the $\Z_8$-linear codes.
\end{abstract}


\section{Introduction}\label{sec:intro}

Let $\Z_{2^s}$ be the ring of integers modulo $2^s$ with $s\geq1$.  The set
of $n$-tuples over $\Z_{2^s}$ is denoted by $\Z_{2^s}^n$, and in this paper
the elements of $\Z_{2^s}^n$ are also called vectors. A code over $\Z_2$ of
length $N$ is a nonempty subset of $\Z_2^N$, and it is linear if it is a
subspace of $\Z_2^N$.  Similarly, a nonempty subset of $\Z_{2^s}^n$ is a $\Z_{2^s}$-additive code if it is a subgroup of the additive group of $\Z_{2^s}^n$. A $\Z_2\Z_4\Z_8$-additive code is a subgroup of $\Z_2^{\alpha_1} \times \Z_4^{\alpha_2}\times \Z_8^{\alpha_3}$. Note that a $\Z_2\Z_4\Z_8$-additive code is a linear code over $\Z_2$ when $\alpha_2=\alpha_3=0$,  a $\Z_4$-additive or $\Z_8$-additive code when $\alpha_1=\alpha_3=0$ or $\alpha_1=\alpha_2=0$, respectively, and a $\Z_2\Z_4$-additive code when $\alpha_3=0$. The order of a vector $\uu\in \Z_{2^s}^n$, denoted by $\ord(\uu)$, is the smallest positive integer $m$ such that $m \uu =(0,\dots,0)$. Also, the order of a vector $\uu\in \Z_2^{\alpha_1}\times\Z_4^{\alpha_2} \times\Z_8^{\alpha_3}$, denoted by $\ord(\uu)$, is the smallest positive integer $m$ such that $m \uu =(0,\dots,0\mid 0,\dots,0 \mid 0,\dots,0)$.

Two binary codes $C_1$ and $C_2$ of length $N$ are said to be equivalent if
there are a vector $\uu\in\Z_2^N$ and a permutation of coordinates $\pi$ such
that $C_2=\{\uu+\pi(\xx):\xx\in C_1\}$.  The Hamming weight of a vector
$\uu\in\Z_2^N$, denoted by $\wt_H(\uu)$, is the number of its nonzero
coordinates, and the Hamming distance $d_H(\uu,\vv)$ of two vectors is the
number of coordinates in which they differ; hence,
$d_H(\uu,\vv)=\wt_H(\uu-\vv)$.  The minimum distance of a code $C$ over $\Z_2$
is $d(C)=\min\{d_H(\uu,\vv):\uu,\vv\in C,\ \uu\not=\vv\}$.

In \cite{Sole}, a Gray map from $\Z_4$ to $\Z_2^2$ is defined as
$\phi(0)=(0,0)$, $\phi(1)=(0,1)$, $\phi(2)=(1,1)$ and $\phi(3)=(1,0)$.  There
exist different generalizations of this Gray map, which go from $\Z_{2^s}$ to
$\Z_2^{2^{s-1}}$ \cite{Carlet,Codes2k,dougherty,Nechaev,Krotov:2007}.  In this
paper, we focus on Carlet's Gray map \cite{Carlet}, from $\Z_{2^s}$ to
$\Z_2^{2^{s-1}}$, which is a particular case of the one given in
\cite{Krotov:2007,ShiKrotov2019} satisfying
$\sum\lambda_i\phi_s(2^i)=\phi_s(\sum\lambda_i2^i)$ \cite{KernelZ2s}.
Specifically,
\begin{equation}\label{eq:Carlet}
\phi_s(u)=(u_{s-1},u_{s-1},\dots,u_{s-1})+(u_0,\dots,u_{s-2})Y_{s-1},
\end{equation}
where $u\in\Z_{2^s}$; $[u_0,u_1,\dots,u_{s-1}]_2$ is the binary expansion of
$u$, that is, $u=\sum_{i=0}^{s-1}u_i2^i$ with $u_i\in\{0,1\}$; and $Y_{s-1}$
is a matrix of size $(s-1)\times2^{s-1}$ whose columns are all the vectors in
$\Z_2^{s-1}$.  Without loss of generality, we assume that the columns of
$Y_{s-1}$ are ordered in ascending order, by considering the elements of
$\Z_2^{s-1}$ as the binary expansions of the elements of $\Z_{2^{s-1}}$. Note that $\phi_1$ is the identity map and that $\phi_2$ is the Gray map $\phi$ given above.  We define
$\Phi_s:\Z_{2^s}^n\rightarrow\Z_2^{n2^{s-1}}$ as the component-wise extended
map of $\phi_s$, and a Gray map $\Phi$ from
$\Z_2^{\alpha_1}\times\Z_4^{\alpha_2}\times\Z_8^{\alpha_3}$ to $\Z_2^N$, where
$N=\alpha_1+2\alpha_2+4\alpha_3$, by
$$
\Phi(\uu_1\mid \uu_2\mid \uu_3)=(\uu_1,\Phi_2(\uu_2),\Phi_3(\uu_3)),
$$
for any $\uu_i\in\Z_{2^i}^{\alpha_i}$, where $1\leq i\leq3$.

Let $\C\subseteq\Z_{2^s}^n$ be a $\Z_{2^s}$-additive code of length $n$.  We
say that the Gray map image of $\C$, say $C=\Phi_s(\C)$, is a
$\Z_{2^s}$-linear code of length $n2^{s-1}$.  Since $\C$ is a subgroup of
$\Z_{2^s}^n$, it is isomorphic to
$\Z_{2^s}^{t_1}\times\Z_{2^{s-1}}^{t_2}\times\dots\times\Z_2^{t_s}$, and we
say that $\C$, or equivalently $C=\Phi_s(\C)$, is of type $(n;t_1,\dots,t_s)$.
Similarly, if
$\C\subseteq\Z_2^{\alpha_1}\times\Z_4^{\alpha_2}\times\Z_8^{\alpha_3}$ is a
$\Z_2\Z_4\Z_8$-additive code, we say that its Gray map image $C=\Phi(\C)$ is a
$\Z_2\Z_4\Z_8$-linear code of length $\alpha_1+2\alpha_2+4\alpha_3$.  Since
$\C$ can be seen as a subgroup of $\Z_8^{\alpha_1+\alpha_2+\alpha_3}$, it is
isomorphic to $\Z_8^{t_1}\times\Z_4^{t_2}\times\Z_2^{t_3}$, and we say that
$\C$, or equivalently $C=\Phi(\C)$, is of type
$(\alpha_1,\alpha_2,\alpha_3;t_1,t_2,t_3)$.  Note that a $\Z_2\Z_4$-linear
code $\C$ \cite{ccsg,BookZ2Z4} can be seen as a $\Z_2\Z_4\Z_8$-linear code of
type $(\alpha_1,\alpha_2,0;0,t_2,t_3)$.  In this case, we also say that the
type of $\C$ is directly $(\alpha_1,\alpha_2;t_2,t_3)$.  Unlike linear codes
over finite fields, linear codes over rings do not have a basis, but there
exists a generator matrix for these codes having minimum number of rows.  If
$\C$ is a $\Z_2\Z_4\Z_8$-additive code of type
$(\alpha_1,\alpha_2,\alpha_3;t_1,t_2,t_3)$, then $|\C|=8^{t_1}4^{t_2}2^{t_3}$
and there exists a generator matrix with $t_1+t_2+t_3$ rows.

A binary code of length $N$, $2N$ codewords and minimum distance $N/2$ is called a Hadamard code. Hadamard codes can be constructed from Hadamard matrices \cite{Key,WMcwill}. Note that linear Hadamard codes are in fact first order Reed-Muller codes, or equivalently, the dual of extended Hamming codes \cite{WMcwill}. It is also important to note that Hadamard codes are two weight codes, which have been widely studied in \cite{ShiTwoHomWeight,TwoWeightSole}.
 Most Hadamard codes are, however, nonlinear, and their classification is still an open problem.  One fruitful way
of attacking it is to realise some of them as Gray map images of additive
codes over rings or over mixed alphabets and then to decide which of the
resulting codes are equivalent, each such result giving a partial
classification of nonlinear Hadamard codes.  The $\Z_4$ representation of the
Kerdock, Preparata and Goethals codes \cite{Sole} was the starting point of
this point of view, and codes over $\Z_{p^s}$ go back to Blake \cite{Blake}
and Shankar \cite{Shankar}. 

From a more practical point of view, since Hadamard codes are optimal and have a high correction capability, they appear in different aspects related to the transmission of information, such as in digital communication with satellites \cite{H07}, in CDMA phones to modulate the transmission of information and minimize interference with other transmissions \cite{SBT98} and, in general, in different OCDMA multiple access systems to allow access to multiple users asynchronously and simultaneously \cite{HYT04}. Other applications are found in cryptography \cite{Nyb91} or in information hiding (steganography and watermarking) \cite{YLL03}. See \cite{H07} for more applications in other fields.

Two structural properties of codes over $\Z_2$ are the rank and the 
dimension of the kernel. The rank of a code $C$ over $\Z_2$ is simply the
dimension of the linear span, $\langle C \rangle$,  of $C$.
The kernel of a code $C$ over $\Z_2$ is defined as
$\K(C)=\{\xx\in \Z_2^N : \xx+C=C \}$ \cite{BGH83}. If the all-zero vector belongs to $C$,
then $\K(C)$ is a linear subcode of $C$.
Note also that if $C$ is linear, then $K(C)=C=\langle C \rangle$.
We denote the rank of $C$ as $\rank(C)$ and the dimension of the kernel as $\kernel(C)$.
Both are equivalence invariants, so two codes with different pairs
$(r,k)=(\rank(C),\kernel(C))$ are nonequivalent; the converse is false in
general, and it is exactly this that makes classification difficult.

The $\Z_{2^s}$-additive codes such that after the Gray map $\Phi_s$ give
Hadamard codes are called $\Z_{2^s}$-additive Hadamard codes and the
corresponding images are called $\Z_{2^s}$-linear Hadamard codes.  Similarly,
the $\Z_2\Z_4\Z_8$-additive codes such that after the Gray map $\Phi$ give
Hadamard codes are called $\Z_2\Z_4\Z_8$-additive Hadamard codes and the
corresponding images are called $\Z_2\Z_4\Z_8$-linear Hadamard codes.  It is
known that $\Z_4$-linear Hadamard codes, that is, $\Z_2\Z_4$-linear Hadamard
codes with $\alpha_1=0$, and $\Z_2\Z_4$-linear Hadamard codes with
$\alpha_1\not=0$ can be classified by using either the rank or the dimension
of the kernel \cite{Kro:2001:Z4_Had_Perf,PRV06}.  Moreover, in \cite{KV2015},
it is shown that each $\Z_4$-linear Hadamard code is equivalent to a
$\Z_2\Z_4$-linear Hadamard code with $\alpha_1\not=0$.  Later, in
\cite{KernelZ2s,HadamardZps,EquivZ2s,ZpsEquivalance}, a recursive construction
for $\Z_{p^s}$-linear Hadamard codes, with $p$ prime, is described, the
linearity is established, and a partial classification by using the dimension
of the kernel is obtained, giving the exact amount of nonequivalent such codes
for some parameters.  In \cite{fernandez2019mathbb}, a complete classification
of $\Z_8$-linear Hadamard codes by using the rank and dimension of the kernel
is provided, giving the exact amount of nonequivalent such codes.  For any
$t\geq2$, the full classification of $\Z_p\Z_{p^2}$-linear Hadamard codes of
length $p^t$, with $\alpha_1\not=0$, $\alpha_2\not=0$, and $p\geq3$ prime, is
given in \cite{ZpZp2Construction,ZpZp2Classification}, by using just the
dimension of the kernel.

This paper is concerned with the mixed alphabet
$\Z_2^{\alpha_1}\times\Z_4^{\alpha_2}\times\Z_8^{\alpha_3}$ with
$\alpha_1\not=0$, $\alpha_2\not=0$ and $\alpha_3\not=0$.  A recursive
construction of $\Z_2\Z_4\Z_8$-additive Hadamard codes $\cH^{t_1,t_2,t_3}$ of
type $(\alpha_1,\alpha_2,\alpha_3;t_1,t_2,t_3)$, with $t_1\geq1$, $t_2\geq0$
and $t_3\geq1$, was given in \cite{Z2Z4Z8Construction}, where it was also
proved that several other natural recursive constructions produce permutation
equivalent codes. The corresponding binary codes are denoted by
$H^{t_1,t_2,t_3}=\Phi(\cH^{t_1,t_2,t_3})$. In \cite{Z2Z4Z8Linearity}, it was
shown that $H^{t_1,t_2,t_3}$ is linear if and only if $(t_1,t_2)=(1,0)$, and
that for the nonlinear ones the kernel is the Gray map image of the subcode of
the elements of order at most two, and hence
$\kernel(H^{t_1,t_2,t_3})=t_1+t_2+t_3$.

Those results give a complete classification for the lengths $2^t$ with $3\leq t\leq11$. However, no general formula for the rank of $H^{t_1,t_2,t_3}$ was known. The rank values used in \cite{Z2Z4Z8Linearity} were computed with {\sc Magma} \cite{Magma}, type by type, for these lengths. Consequently, although the classification in that range was complete, there was no general description of the rank for arbitrary type, nor of which distinct types of the same length share both the rank and the dimension of the kernel. In particular, when two codes had the same values of these two invariants, their nonequivalence was established by computational equivalence tests. Thus, a general rank formula is needed to understand theoretically how far the rank and the dimension of the kernel alone distinguish $\Z_2\Z_4\Z_8$-linear Hadamard codes for arbitrary length, and to identify precisely the cases in which an additional invariant is required.

The first and main contribution of this paper is the explicit determination of
the rank of $H^{t_1,t_2,t_3}$, given in Theorem \ref{thm:rank-main}.  The
proof is elementary but requires care.  Since the row rank of a matrix equals
its column rank, the rank of the code is the dimension of the space spanned by
its binary coordinate functions, each of which is a Boolean function of the
binary digits of the message.  The coordinate functions attached to the
coordinates in the $\Z_2$ and the $\Z_4$ parts span a space $V_0$, which
contains the digit functions and the quadratic carry functions.  The lowest
two binary digits of every coordinate in the $\Z_8$ part already lie in $V_0$;
only the top digit can increase the rank, and we describe the space it
contributes by an explicit and short list of generators.  Every value produced
by the formula agrees with the values obtained with {\sc Magma} in
\cite[Tables 2--4]{Z2Z4Z8Linearity}.

The second contribution is the classification that the rank and the dimension
of the kernel produce, together with a precise description of what they cannot
do.  Inside the family we prove, in Theorem \ref{thm:all-collisions}, that the
only pairs of distinct types with the same length, the same rank and the same
dimension of the kernel are $\bigl(H^{1,2,t-6},\,H^{2,0,t-5}\bigr)$, where
$t\geq7$, and $\bigl(H^{2,2,t-9},\,H^{3,0,t-8}\bigr)$, where $t\geq10$.
Consequently, writing $\cA_t=\lfloor(t^2+6)/12\rfloor$ for the number of
admissible types of length $2^t$, the pair $(r,k)$ separates exactly $\cA_t$
classes for $3\leq t\leq6$, then $\cA_t-1$ for $7\leq t\leq9$, and then
$\cA_t-2$ for every $t\geq10$.  In particular, $(r,k)$ classifies the whole
family if and only if $3\leq t\leq6$.  This is Corollary
\ref{cor:rk-classes}, and it is the exact statement of what remains to be done
by other means.

The third contribution is the comparison with the classical families by means
of the same two invariants.  We determine, in Theorem
\ref{thm:z2z4-coincidences}, all the coincidences with the $\Z_4$-linear and
the $\Z_2\Z_4$-linear Hadamard codes: there is exactly one infinite family,
namely $H^{2,1,t-7}$ and  $H^{5,t-9}$, where  $t\geq 10,$ and for every other type the rank or the dimension of the kernel already separates.  We also determine, in Proposition \ref{prop:z8-coincidences}, all the coincidences with the $\Z_8$-linear Hadamard codes; there are exactly five families. 

The paper is organised as follows. Section \ref{sec:preliminaries} contains the preliminaries. We recall the construction of $\Z_2\Z_4\Z_8$-linear Hadamard codes of type $(\alpha_1,\alpha_2,\alpha_3;t_1,t_2,t_3)$, with $\alpha_1\neq 0$, $\alpha_2\neq 0$, and $\alpha_3\neq 0$, given in \cite{Z2Z4Z8Construction}, together with a direct description of the columns of their generator matrices $A^{t_1,t_2,t_3}$, which replaces the recursive construction by a closed combinatorial statement and is the form used in all subsequent proofs. We also recall the linearity and kernel results given in \cite{Z2Z4Z8Linearity}. Section \ref{sec:rank} develops the Boolean function method and proves the rank formula.  Section \ref{sec:rk} carries out the classification by the rank and the dimension of the kernel, inside the family and against the classical families, and states precisely which pairs
remain unseparated.  Section \ref{sec:conclusions} summarises the results.

\section{Preliminaries and the recursive family}\label{sec:preliminaries}

In this section, we recall the recursive construction of the matrices
$A^{t_1,t_2,t_3}$ given in \cite{Z2Z4Z8Construction}.  We then replace the
recursion by a direct description of their columns, which is the form used
throughout the proofs of Section \ref{sec:rank}.  Finally, we recall the
linearity and the kernel of $H^{t_1,t_2,t_3}$ obtained in
\cite{Z2Z4Z8Linearity}.

\subsection{The recursive construction of $A^{t_1,t_2,t_3}$}
\label{subsec:construction}

We now recall the construction of \cite{Z2Z4Z8Construction}. Let $\zero,\one,\two,\dots,\mathbf{7}$ denote the vectors having the elements $0,1,2,\dots,7$ repeated in every coordinate, the length being always clear from the context.  Let $t_1\geq1$, $t_2\geq0$ and $t_3\geq1$ be integers.
The matrices $A^{t_1,t_2,t_3}$, having $t_1$ rows of order $8$, $t_2$ rows of
order $4$ and $t_3$ rows of order $2$, are constructed recursively as
follows.  We start from
\begin{equation}\label{eq:A101}
A^{1,0,1}=\left(\begin{array}{cc|c|c}
1&1&2&4\\
0&1&1&1\\
\end{array}\right),
\end{equation}
and then apply the three constructions below.  If
$A^{\ell-1,0,1}=(A_1\mid A_2\mid A_3)$ with $\ell\geq2$ is already
constructed, we form
\begin{equation}\label{eq:construction-order8}
A^{\ell,0,1}=\left(\begin{array}{cc|ccccc|ccccc}
A_1&A_1&M_1&A_2&A_2&A_2&A_2&M_2&A_3&A_3&\cdots&A_3\\
\zero&\one&\one&\zero&\one&\two&\mathbf{3}&\one&\zero&\one&\cdots&\mathbf{7}
\end{array}\right),
\end{equation}
where $M_1=\{\zz^T:\zz\in\{2\}\times\{0,2\}^{\ell-1}\}$ and
$M_2=\{\zz^T:\zz\in\{4\}\times\{0,2,4,6\}^{\ell-1}\}$.  Construction
\eqref{eq:construction-order8} is applied until $\ell=t_1$.  If
$A^{t_1,\ell-1,1}=(A_1\mid A_2\mid A_3)$ with $t_1\geq1$ and $\ell\geq1$ is
already constructed, we form
\begin{equation}\label{eq:construction-order4}
A^{t_1,\ell,1}=\left(\begin{array}{cc|ccccc|cccc}
A_1&A_1&M_1&A_2&A_2&A_2&A_2&A_3&A_3&A_3&A_3\\
\zero&\one&\one&\zero&\one&\two&\mathbf{3}&\zero&\two&\four&\six
\end{array}\right),
\end{equation}
where now $M_1=\{\zz^T:\zz\in\{2\}\times\{0,2\}^{t_1+\ell-1}\}$, and
Construction \eqref{eq:construction-order4} is applied until $\ell=t_2$.
Finally, if $A^{t_1,t_2,\ell-1}=(A_1\mid A_2\mid A_3)$ with $\ell\geq2$ is
already constructed, we form
\begin{equation}\label{eq:construction-order2}
A^{t_1,t_2,\ell}=\left(\begin{array}{cc|cc|cc}
A_1&A_1&A_2&A_2&A_3&A_3\\
\zero&\one&\zero&\two&\zero&\four
\end{array}\right),
\end{equation}
and Construction \eqref{eq:construction-order2} is applied until
$\ell=t_3$. Thus, in this way, we obtain $A^{t_1,t_2,t_3}$.

Summarising, in order to reach $A^{t_1,t_2,t_3}$ from $A^{1,0,1}$ we first
adds $t_1-1$ rows of order $8$ by applying \eqref{eq:construction-order8}
exactly $t_1-1$ times, then add $t_2$ rows of order $4$ by applying
\eqref{eq:construction-order4} exactly $t_2$ times, and finally add $t_3-1$
rows of order $2$ by applying \eqref{eq:construction-order2} exactly $t_3-1$
times.  Note that the first row of $A^{t_1,t_2,t_3}$ is always
$(\one\mid\two\mid\four)$, and that it has order $2$; it is therefore one of
the $t_3$ rows of order two, and it will play a distinguished role
throughout.  We denote by $\cH^{t_1,t_2,t_3}$ the $\Z_2\Z_4\Z_8$-additive
code generated by $A^{t_1,t_2,t_3}$ and by
$H^{t_1,t_2,t_3}=\Phi(\cH^{t_1,t_2,t_3})$ the corresponding
$\Z_2\Z_4\Z_8$-linear code.

\begin{example}\label{ex:matricesA}
By using the constructions described in (\ref{eq:construction-order8}), (\ref{eq:construction-order4}), and (\ref{eq:construction-order2}), we obtain the following matrices $A^{2,0,1}$, $A^{1,1,1}$ and $A^{1,1,2}$, respectively, starting from $A^{1,0,1}$ given in (\ref{eq:A101}):
\begin{equation}\label{eq:A201}
A^{2,0,1}=\left(\begin{array}{cc|cc|cc}
11&11&22&2222&4444&44444444\\
01&01&02&1111&0246&11111111\\
00&11&11&0123&1111&01234567
\end{array}\right),
\end{equation}
\begin{equation}\label{eq:A111-reordered}
A^{1,1,1}=
\left(\begin{array}{cc|cc|c}
 11 &11 &22&2222 &4444\\
 01&01 &02&1111 &1111\\
 00&11 &11&0123 &0246\\ 
\end{array}\right), 
\end{equation}
$$
A^{1,1,2}=
\left(\begin{array}{cc|cc|cc}
 1111 & 1111 &222222 &222222 &4444 &4444\\
 0101 & 0101 &021111 &021111 &1111 &1111\\
 0011 & 0011 &110123 &110123 &0246 &0246\\ 
 0000 &1111  &000000 &222222 &0000 &4444\\
\end{array}\right).
$$
\end{example}

\subsection{A direct description of the columns of $A^{t_1,t_2,t_3}$}
\label{subsec:columns}

The recursion defines the family, but for the proofs it is far more convenient
to know directly which columns occur.  We give such a description and check
that it agrees with the recursion.  It is stated for $t_3=1$, since the
remaining rows of order two are added by the duplication
\eqref{eq:construction-order2}.  Recall that an element of $\Z_{2^s}$ is a
unit if and only if it is odd.

\begin{definition}\label{def:normalized-primitive}
Let $\zz=(z_1,\dots,z_\ell)\in\Z_{2^s}^\ell$.  We say that $\zz$ is
\emph{primitive} if at least one of its entries is odd, or equivalently if
its entries do not all lie in the ideal $2\Z_{2^s}$.  If $\zz$ is primitive
and $p=\min\{i:z_i\text{ is odd}\}$ is the position of its first odd entry,
we say that $\zz$ is \emph{normalized primitive} when $z_p=1$.
\end{definition}

Normalization selects exactly one vector from each orbit under multiplication
by a unit: if the first odd entry of $\zz$ is $z_p$, then $z_p^{-1}\zz$ is
normalized primitive, and this representative is unique because $rz_p=r'z_p=1$
forces $r=r'=z_p^{-1}$.  For instance $(3,2)\in\Z_4^2$ is primitive but not
normalized; its first odd entry is $3$ and, since $3^{-1}=3$ in $\Z_4$, its normalized representative is $(1,2)$. The six normalized primitive vectors of $\Z_4^2$ are $(1,0)$, $(1,1)$, $(1,2)$, $(1,3)$, $(0,1)$ and $(2,1)$.  Note that normalization fixes the entry at the position $p$ only: neither the even entries before $p$ nor the entries after $p$ are forced to vanish.

\begin{proposition}\label{prop:column-description}
Let $t_1\geq1$ and $t_2\geq0$. Delete from every column of $A^{t_1,t_2,1}$ its first entry, which is $1$, $2$ or $4$ according to whether the additive coordinate lies in the $\Z_2$, the $\Z_4$ or the $\Z_8$ part. Then, the sets of the remaining truncated columns are the following.
\begin{enumerate}[label=\textup{(\roman*)}]
\item The $\Z_2$ part consists of every vector of $\Z_2^{t_1+t_2}$, each one
occurring exactly once.
\item The $\Z_4$ part consists of every normalized primitive vector of
$\Z_4^{t_1+t_2}$, each one occurring exactly once.
\item The $\Z_8$ part consists of every vector $(q_1,\dots,q_{t_1},2r_1,\dots,2r_{t_2})$, where $(q_1,\dots,q_{t_1})$ is normalized primitive in $\Z_8^{t_1}$ and
$(r_1,\dots,r_{t_2})\in\Z_4^{t_2}$ is arbitrary, each one occurring exactly
once.
\end{enumerate}
For $t_3>1$, the matrix $A^{t_1,t_2,t_3}$ is obtained from $A^{t_1,t_2,1}$ by
applying the duplication \eqref{eq:construction-order2} exactly $t_3-1$
times.
\end{proposition}

\begin{proof}
For $A^{1,0,1}$, the three assertions follow directly from
\eqref{eq:A101}: the $\Z_2$ part is $\{(0),(1)\}=\Z_2^1$, the $\Z_4$ part is
$\{(1)\}$, which is the unique normalized primitive vector of $\Z_4^1$, and
the $\Z_8$ part is $\{(1)\}$, which is the unique normalized primitive vector
of $\Z_8^1$.

Suppose that the description holds for $A^{\ell-1,0,1}$ and that a new row of
order $8$ is added by \eqref{eq:construction-order8}.  In the $\Z_2$ part,
every old truncated column occurs twice, once with new entry $0$ and once with
new entry $1$, so all the vectors of $\Z_2^{\ell}$ occur exactly once.  In the
$\Z_4$ part, the four repetitions of $A_2$ carry the new entries
$\zero,\one,\two,\mathbf{3}$, so every old normalized primitive vector is
extended by all four possible new entries; these are precisely the normalized
primitive vectors of $\Z_4^{\ell}$ whose first odd entry occurs among the
first $\ell-1$ positions.  The extra block $M_1$ contributes the columns whose
first $\ell-1$ entries lie in $\{0,2\}$ and whose new entry is $1$, that is,
precisely those whose first odd entry occurs in the new position.  Hence, every
normalized primitive vector of $\Z_4^{\ell}$ occurs exactly once, and no other
primitive vector occurs.  The same argument applies verbatim to the $\Z_8$
part, with $A_3$ repeated eight times with new entries
$\zero,\one,\dots,\mathbf{7}$ and with the extra block $M_2$, whose entries in
the first $\ell-1$ positions lie in $\{0,2,4,6\}$ and whose new entry is $1$.

Suppose now that the description holds for $A^{t_1,\ell-1,1}$ and that a new
row of order $4$ is added by \eqref{eq:construction-order4}.  In the $\Z_2$
and the $\Z_4$ parts, the argument is unchanged.  In the $\Z_8$ part, a row of
order $4$ can only have entries of order at most $4$, that is, entries in
$\{0,2,4,6\}$, and the four repetitions of $A_3$ carry exactly the new entries
$\zero,\two,\four,\six$; so those columns are extended by all the values $2r$
with $r\in\Z_4$, which is item (iii).  The two inductions, on $t_1$ and then
on $t_2$, prove the three assertions. The last statement of the
proposition is Construction \eqref{eq:construction-order2} itself.
\end{proof}

The type of the codes now follows by counting the coordinates of each kind.

\begin{proposition}[\cite{Z2Z4Z8Construction}]\label{prop:coordinate-counts}
Let $t_1\geq1$, $t_2\geq0$ and $t_3\geq1$, and let the type of
$\cH^{t_1,t_2,t_3}$ be $(\alpha_1,\alpha_2,\alpha_3;t_1,t_2,t_3)$.  Then,
$\alpha_1=2^{t_1+t_2+t_3-1}$,
$\alpha_1+2\alpha_2=4^{t_1+t_2}2^{t_3-1}$, and
$\alpha_1+2\alpha_2+4\alpha_3=8^{t_1}4^{t_2}2^{t_3-1}$.
In particular, $\alpha_1\not=0$, $\alpha_2\not=0$ and $\alpha_3\not=0$, the
binary length of $H^{t_1,t_2,t_3}$ is $N=2^t$ with
\begin{equation}\label{eq:length-relation}
t+1=3t_1+2t_2+t_3,
\end{equation}
and $|H^{t_1,t_2,t_3}|=8^{t_1}4^{t_2}2^{t_3}=2^{t+1}=2N$.
\end{proposition}

We illustrate Proposition \ref{prop:column-description} on the smallest
nonlinear member of the family, which will be used repeatedly.

\begin{example}\label{ex:A111-columns}
Consider the matrix $A^{1,1,1}$, given in \eqref{eq:A111-reordered}. Delete the
first entry of every column.  In the $\Z_2$ part, the four truncated columns
are $(0,0)$, $(1,0)$, $(0,1)$, $(1,1)$, that is, all of $\Z_2^2$; in the
$\Z_4$ part the six truncated columns are $(0,1)$, $(2,1)$, $(1,0)$, $(1,1)$,
$(1,2)$, $(1,3)$, exactly the six normalized primitive vectors of $\Z_4^2$
listed above; and in the $\Z_8$ part the four truncated columns are $(1,2r)$
with $r\in\Z_4$, since $(1)$ is the only normalized primitive vector of
$\Z_8^1$.  This is what items (i), (ii) and (iii) assert.  The type is
$(4,6,4;1,1,1)$, the binary length is $4+2\cdot6+4\cdot4=32=2^5$, in
accordance with $t+1=6=3+2+1$, and $|H^{1,1,1}|=8\cdot4\cdot2=64=2\cdot32$.
\end{example}

\subsection{Linearity and kernel of $H^{t_1,t_2,t_3}$}
\label{subsec:recalled}

We finally recall the following three results of \cite{Z2Z4Z8Construction} and \cite{Z2Z4Z8Linearity}.
\begin{theorem}[\cite{Z2Z4Z8Construction}]\label{thm:recalled-hadamard}
Let $t_1\geq1$, $t_2\geq0$ and $t_3\geq1$ be integers.  Then, the
$\Z_2\Z_4\Z_8$-additive code $\cH^{t_1,t_2,t_3}$ generated by
$A^{t_1,t_2,t_3}$ is a $\Z_2\Z_4\Z_8$-additive Hadamard code; that is,
$H^{t_1,t_2,t_3}=\Phi(\cH^{t_1,t_2,t_3})$ is a binary Hadamard code of length
$2^t$ with $t+1=3t_1+2t_2+t_3$.
\end{theorem}

\begin{theorem}[\cite{Z2Z4Z8Linearity}]\label{thm:recalled-linearity-kernel}
Let $t_1\geq1$, $t_2\geq0$ and $t_3\geq1$ be integers.
\begin{enumerate}[label=\textup{(\roman*)}]
\item The code $H^{t_1,t_2,t_3}$ is linear if and only if $(t_1,t_2)=(1,0)$;
equivalently, the linear members of the family are exactly the codes
$H^{1,0,t_3}$ with $t_3\geq1$.
\item Assume that $H^{t_1,t_2,t_3}$ is nonlinear and let
$\cH^{t_1,t_2,t_3}_2$ be the subcode of $\cH^{t_1,t_2,t_3}$ consisting of all
its elements of order at most two.  Then,
$\K(H^{t_1,t_2,t_3})=\Phi\bigl(\cH^{t_1,t_2,t_3}_2\bigr)$ and $\kernel(H^{t_1,t_2,t_3})=t_1+t_2+t_3$.
\item If $\ww_1,\dots,\ww_{t_1+t_2+t_3}$ are the rows of $A^{t_1,t_2,t_3}$,
in any order, and $o_i=\ord(\ww_i)$, then $\{\Phi(\tfrac{o_i}{2}\ww_i):1\leq
i\leq t_1+t_2+t_3\}$ is a basis of $\K(H^{t_1,t_2,t_3})$.
\end{enumerate}
\end{theorem}

\section{The rank of $H^{t_1,t_2,t_3}$}\label{sec:rank}

In this section, we determine the rank of every code of the family.  The
rank of the $\Z_4$-linear and of the $\Z_2\Z_4$-linear Hadamard codes was
computed in \cite{Kro:2001:Z4_Had_Perf,PRV06}, and the rank of the
$\Z_8$-linear Hadamard codes in \cite{fernandez2019mathbb}; no such formula
was available for the mixed alphabet $\Z_2\Z_4\Z_8$, and the values used in
\cite{Z2Z4Z8Linearity} were obtained type by type with {\sc Magma}
\cite{Magma}.  Here, we obtain a closed formula valid for all $t_1\geq1$,
$t_2\geq0$ and $t_3\geq1$.

Subsection \ref{subsec:rows-columns} fixes the vocabulary of rows, coordinates,
and digits and turns the rank into the dimension of a space of Boolean
functions.  This reduction, recorded in \eqref{eq:rank-as-function-span}, is
the basis of the whole section.  The distinction between rows and coordinates
is essential throughout: the coefficients $x_i,y_j$ multiply generator
\emph{rows}, whereas the parameters $q_i,r_j$ describe \emph{coordinates}. The computation itself has four steps.  We first determine the space $V_0$ spanned by the coordinate functions coming from the coordinates in the $\Z_2$ and $\Z_4$ parts.  We then expand the three binary digits of a general
coordinate in the $\Z_8$ part and show that the bottom two already lie in
$V_0$.  Thus only the top digit can increase the rank, and we reduce this
digit modulo $V_0$ to the short expression of Lemma \ref{lem:top-digit-reduced}.  We then determine the dimension of the span of the reduced functions.  Finally, the duplication \eqref{eq:construction-order2} carries the result from $t_3=1$ to an arbitrary $t_3$.

\subsection{Generator rows, additive coordinates and binary coordinates}
\label{subsec:rows-columns}

Five objects occur in the computations below, and confusing any two of them
would make the arguments impossible to follow: \emph{generator rows},
\emph{additive coordinates}, \emph{message coefficients}, \emph{binary message
digits} and \emph{binary coordinates}.  We fix their names here and use no
others.

Let $A=(A_1\mid A_2\mid A_3)$ be a generator matrix of a
$\Z_2\Z_4\Z_8$-additive code $\C$ of type
$(\alpha_1,\alpha_2,\alpha_3;t_1,t_2,t_3)$, where $A_1$, $A_2$, and $A_3$
correspond to the $\Z_2$-, $\Z_4$-, and $\Z_8$-coordinate positions,
respectively.  The rows of $A$ are the \emph{generator rows}, and the
coordinate positions of $\C$ are its \emph{additive coordinates}. An
additive coordinate lies in the $\Z_2$ part, in the $\Z_4$ part, or in the
$\Z_8$ part according to whether the corresponding column of $A$ belongs to
$A_1$, $A_2$, or $A_3$.  That column lists the values taken at that coordinate
by the generator rows, and we identify an additive coordinate with its
corresponding column whenever we refer to its entries.

Let $\ww_1,\dots,\ww_{t_1}$, $\vv_1,\dots,\vv_{t_2}$, and
$\zz_1,\dots,\zz_{t_3}$ denote the generator rows of $A$ of orders $8$, $4$,
and $2$, respectively.  Then every additive codeword $\uu\in\C$ can be
written uniquely as
\begin{equation}\label{eq:additive-message}
\uu=\sum_{i=1}^{t_1}x_i\ww_i+\sum_{j=1}^{t_2}y_j\vv_j
      +\sum_{k=1}^{t_3}\epsilon_k\zz_k,
\end{equation}
where $x_i\in\Z_8$, $y_j\in\Z_4$, and $ \epsilon_k\in\Z_2$.
The tuple $\xi=(x_1,\dots,x_{t_1},y_1,\dots,y_{t_2},\epsilon_1,\dots,\epsilon_{t_3})$ is
called the \emph{message} of the codeword $\uu$, and its entries
$x_i,y_j,\epsilon_k$ are the \emph{message coefficients}. A message coefficient tells how many times the corresponding generator row is added to itself.  The set of all messages is $\cM=\Z_8^{t_1}\times\Z_4^{t_2}\times\Z_2^{t_3}$, so $|\cM|=|\C|$ and \eqref{eq:additive-message} defines a bijection $\cM\rightarrow\C$. 

For the rank computation, we replace each message coefficient by its binary
digits, writing
\begin{equation}\label{eq:message-digits}
x_i=a_i+2b_i+4c_i\ \ \text{in }\Z_8,
\qquad
y_j=d_j+2e_j\ \ \text{in }\Z_4,
\end{equation}
with $a_i,b_i,c_i,d_j,e_j\in\{0,1\}$.  From that point on, we regard the
$a_i,b_i,c_i$ $(1\leq i\leq t_1)$, the $d_j,e_j$ $(1\leq j\leq t_2)$, and
the $\epsilon_k$ $(1\leq k\leq t_3)$ as independent variables over $\Z_2$;
they are called the \emph{binary message digits}, or simply the message
digits.  Thus, a message digit is one binary digit of one message coefficient,
and $\epsilon_k$ is its own single digit. The same convention applies to every ring element occurring below.  For $w\in\Z_{2^s}$, we always write $w=\sum_{i=0}^{s-1}w_i2^i$ with $w_i\in\{0,1\}$.  We call $w_0$ the \emph{bottom digit} and $w_{s-1}$ the \emph{top digit} of $w$ and, when $s=3$, we call $w_1$ the \emph{middle digit}.  The digits are extracted inside $\Z_{2^s}$ but are treated as elements of $\Z_2$ afterwards, since by \eqref{eq:Carlet} the Gray image is obtained from them by linear algebra over $\Z_2$, and the rank is the dimension of a $\Z_2$-span.

A position of $\Phi(\C)$ is called a \emph{binary coordinate}.  The two kinds
of position are always distinguished in the same way: an additive coordinate
is named together with its part, as in ``a coordinate in the $\Z_4$ part'',
whereas ``binary coordinate'', or simply ``coordinate'', always means a
position of the binary code.

We isolate next the consequence of \eqref{eq:Carlet} that will be used
repeatedly below: at a given additive coordinate, the binary coordinate
functions produced by the Gray map span exactly the same space as the binary
digits of that coordinate.  In our applications, however, the value at such
a coordinate is not a fixed element of $\Z_{2^s}$, since it varies with the
message.  We therefore write this value as $u(\omega)$, where $\omega$ runs
over a finite parameter set $\Omega$, and apply $\phi_s$ to $u(\omega)$ for
each $\omega\in\Omega$.  As $\omega$ varies, each binary coordinate of the
resulting Gray image defines a function on $\Omega$.

\begin{lemma}\label{lem:gray-digits}
Let $s\geq1$, let $\Omega$ be a finite set, and let
$u:\Omega\rightarrow\Z_{2^s}$ be any map.  For $0\leq i\leq s-1$, let
$u_i:\Omega\rightarrow\Z_2$ be the $i$th binary digit of $u$, that is, 
$u(\omega)=\sum_{i=0}^{s-1}u_i(\omega)2^i$ in $\Z_{2^s}$ for every $\omega\in\Omega$.  Let
$g_1,\dots,g_{2^{s-1}}:\Omega\rightarrow\Z_2$ be defined by
$$
\phi_s\bigl(u(\omega)\bigr)
=
\bigl(g_1(\omega),\dots,g_{2^{s-1}}(\omega)\bigr),
$$
for every $\omega\in\Omega$. Then, the functions $g_1,\dots,g_{2^{s-1}}$ and $u_0,u_1,\dots,u_{s-1}$ span the same $\Z_2$-subspace of the space of all maps $\Omega\rightarrow\Z_2$.
\end{lemma}

\begin{proof}
Fix $\omega\in\Omega$ and apply \eqref{eq:Carlet} to $u(\omega)$.  If the
$j$th column of $Y_{s-1}$ is $\tau=(\tau_0,\dots,\tau_{s-2})$, then the $j$th
coordinate of $\phi_s(u(\omega))$ is
$u_{s-1}(\omega)+\sum_{i=0}^{s-2}\tau_iu_i(\omega)$.  Since this equality
holds for every $\omega\in\Omega$, we obtain the equality of functions
\begin{equation}\label{eq:gray-coordinate-function}
g_j=u_{s-1}+\sum_{i=0}^{s-2}\tau_iu_i .
\end{equation}
Hence, every $g_j$ belongs to the span of
$u_0,\dots,u_{s-1}$. Conversely, the zero column of $Y_{s-1}$ gives the function
$u_{s-1}$.  For each $0\leq i\leq s-2$, the column equal to the $i$th
standard basis vector gives $u_{s-1}+u_i$, and hence $u_i$ also belongs to
the span of the functions $g_j$.  Therefore, the two families of functions
span the same $\Z_2$-subspace.  When $s=1$, we simply have $g_1=u_0$.
\end{proof}
\begin{example}
For $s=3$, the columns of $Y_2=\left(\begin{smallmatrix}0&1&0&1\\0&0&1&1\end{smallmatrix}\right)$ are
$(0,0)$, $(1,0)$, $(0,1)$, and $(1,1)$.  Hence, \eqref{eq:gray-coordinate-function} gives
$g_1=u_2$, $g_2=u_0+u_2$, $g_3=u_1+u_2$, and $g_4=u_0+u_1+u_2$.  Thus, the four binary coordinate functions span the same space as $u_0,u_1,u_2$.
\end{example}
We now introduce the matrix that turns the rank into a question about
functions.  For a message $\xi\in\cM$, let $\uu(\xi)\in\C$ be the additive
codeword given by \eqref{eq:additive-message}. Fix any ordering
$\xi^{(1)},\dots,\xi^{(|\cM|)}$ of $\cM$, and let $E$ be the binary matrix
whose $i$th row is $\Phi(\uu(\xi^{(i)}))$. Thus, the rows of $E$ are all the
codewords of $\Phi(\C)$. We call $E$ the \emph{binary evaluation matrix} of
$\C$.

For a binary coordinate $\ell$, the $\ell$th column of $E$ is the
vector of values of the map $f_\ell:\cM\rightarrow\Z_2$ given by
$f_\ell(\xi)=(\Phi(\uu(\xi)))_\ell$, which we call the \emph{binary coordinate
function} of that position.  Since digit extraction identifies $\cM$, as a
set, with $\Z_2^{3t_1+2t_2+t_3}$, each $f_\ell$ is a Boolean function of the
message digits, nonlinear in general. These nonlinear terms are the source of the
additional rank computed below.

We now describe explicitly the binary coordinate functions for the code $\C$.
Put $m=t_1+t_2+t_3$, and consider an additive coordinate whose corresponding
column has entries $z_1,\dots,z_m$, listed in the same order as the generator
rows.  Writing a message as $\xi=(\xi_1,\dots,\xi_m)$, the codeword $\uu(\xi)$
takes at this coordinate the ring element $\sum_{i=1}^{m}\xi_iz_i$, computed
in $\Z_2$, $\Z_4$ or $\Z_8$ according to the part in which the coordinate
lies.  Its binary digits are Boolean functions of the message digits.  By
\eqref{eq:Carlet}, the one, two or four binary coordinate functions produced
by the Gray map are $\Z_2$-linear combinations of these digit functions, and
by Lemma \ref{lem:gray-digits}, they span the same space as those digits.
Briefly, a binary coordinate selects the corresponding coordinate function,
whereas changing the message moves the point at which it is evaluated.

\begin{lemma}\label{lem:evaluation-principle}
Let $\Omega$ be a finite set, let $f_1,\dots,f_N:\Omega\rightarrow\Z_2$, fix
an ordering $\omega_1,\dots,\omega_{|\Omega|}$ of $\Omega$, and let
$C_\Omega\subseteq\Z_2^N$ be the set of rows of the $|\Omega|\times N$ binary
matrix $E_\Omega=(f_j(\omega_i))_{i,j}$.  Then,
\begin{equation}\label{eq:evaluation-rank}
\rank(C_\Omega)=\dim\Span\{f_1,\dots,f_N\}.
\end{equation}
\end{lemma}

\begin{proof}
The rows of $E_\Omega$ are exactly the elements of $C_\Omega$, possibly with
repetitions, so the row space of $E_\Omega$ is $\langle C_\Omega\rangle$ and
$\rank(C_\Omega)$ is the row rank of $E_\Omega$.  Over any field, the row rank
equals the column rank. The columns of $E_\Omega$ are the evaluation
vectors $\operatorname{ev}(f_j)=(f_j(\omega_1),\dots,f_j(\omega_{|\Omega|}))$.
The map
$$
\operatorname{ev}:
\{f:\Omega\rightarrow\Z_2\}\longrightarrow\Z_2^{|\Omega|},
\qquad
f\longmapsto\bigl(f(\omega_1),\dots,f(\omega_{|\Omega|})\bigr),
$$
is a $\Z_2$-linear isomorphism: it is injective because two functions taking
the same value at every point of $\Omega$ are equal, and surjective because
every vector in $\Z_2^{|\Omega|}$ defines a function on $\Omega$. Hence, the span of the $f_j$ and the span of their evaluation
vectors have the same dimension, and the latter is the column rank of
$E_\Omega$.
\end{proof}

Applying Lemma \ref{lem:evaluation-principle} with $\Omega=\cM$, we obtain
\begin{equation}\label{eq:rank-as-function-span}
\rank(\Phi(\C))
=\dim_{\Z_2}\Span\{f_\ell:\ 1\leq\ell\leq N\},
\end{equation}
the identity on which the rest of this section rests.  It is what makes the
computation feasible.  In the present Hadamard case, the matrix $E$ has
$2^{\,t+1}$ rows and $2^t$ columns, so handling it directly is impractical in
general.  By contrast, \eqref{eq:rank-as-function-span} allows us to work with
the coordinate functions $f_\ell$, each of which has a short explicit
description. 

We close the subsection by fixing the notation used in the rest of the
section.  Throughout, we work with $H^{t_1,t_2,1}$ and pass to general $t_3$
only at the very end.  We denote by $\ww_0=(\one\mid\two\mid\four)$ the
distinguished row of order $2$, by $\ww_1,\dots,\ww_{t_1}$ the rows of order
$8$ and by $\vv_1,\dots,\vv_{t_2}$ the rows of order $4$ of $A^{t_1,t_2,1}$.
Every additive codeword is then written uniquely as
\begin{equation}\label{eq:general-message}
\epsilon\ww_0+\sum_{i=1}^{t_1}x_i\ww_i+\sum_{j=1}^{t_2}y_j\vv_j,
\end{equation}
where  $\epsilon\in\Z_2$, $x_i\in\Z_8$, and  $y_j\in\Z_4$. As in \eqref{eq:message-digits}, we write $x_i=a_i+2b_i+4c_i$ in $\Z_8$
and $y_j=d_j+2e_j$ in $\Z_4$ and regard afterwards
$\epsilon,a_i,b_i,c_i,d_j,e_j$ as independent variables over $\Z_2$.  We will
constantly use the set
\begin{equation}\label{eq:setL}
L=\{a_1,\dots,a_{t_1}\}\cup\{d_1,\dots,d_{t_2}\}
\end{equation}
of the bottom binary digits of the message, which has $t_1+t_2$ elements.

We also recall the standard vocabulary of Boolean functions.  If
$z_1,\dots,z_h$ are independent binary variables, then every function
$f:\Z_2^h\rightarrow\Z_2$ has a unique expression
$f(z_1,\dots,z_h)=\sum_{I\subseteq\{1,\dots,h\}}\gamma_I\prod_{i\in I}z_i$,
 with $\gamma_I\in\Z_2$; this is its algebraic normal form. It is
square-free because $z_i^2=z_i$ over $\Z_2$, and the empty product is the
constant function $1$.  The coefficients are recovered by
$\gamma_I=\sum_{J\subseteq I}f(\one_J)$, where $\one_J$ is the indicator
vector of $J$, so the expression is indeed unique and, in particular,
\emph{distinct square-free monomials are linearly independent as Boolean
functions}.  This last fact justifies every dimension count below.  For a list
$\zz=(z_1,\dots,z_h)$ of binary variables and for $j\geq1$, we write
$\sym_j(\zz)=\sum_{i_1<\cdots<i_j}z_{i_1}\cdots z_{i_j}$ for its $j$th
elementary symmetric Boolean polynomial. Thus, $\sym_1(\zz)=z_1+\cdots+z_h$
and $\sym_2(\zz)=\sum_{i<k}z_iz_k$. When the list is indexed by a subset $S$, we
abbreviate $\sym_j(a_S)=\sym_j\bigl((a_i)_{i\in S}\bigr),$ and similarly for $b$ and $d$.

\subsection{The space $V_0$ coming from the $\Z_2$ and the $\Z_4$ parts}
\label{subsec:V0}

We start with the coordinates over the two smaller alphabets.  Each of them
carries one or two binary coordinates, and the corresponding coordinate
functions span a space $V_0$ consisting of the linear functions of the two
lowest layers of message digits together with the products of two distinct
elements of $L$; those products are the classical carries produced when two
bottom digits are added.

\begin{lemma}\label{lem:V0}
Let $V_0$ denote the $\Z_2$-span of the binary coordinate functions of
$H^{t_1,t_2,1}$ coming from the coordinates in the $\Z_2$ part or in the $\Z_4$
part.  Then, $V_0$ has basis
\begin{equation}\label{eq:V0-basis}
\{\epsilon\}\cup\{a_i,b_i:1\leq i\leq t_1\}
\cup\{d_j,e_j:1\leq j\leq t_2\}
\cup\{uv:\ u,v\in L,\ u\not=v\},
\end{equation}
where each product $uv$ is listed once.  Consequently,
\begin{equation}\label{eq:V0-dim}
\dim V_0=1+2(t_1+t_2)+\binom{t_1+t_2}{2}.
\end{equation}
\end{lemma}

\begin{proof}
By item (i) of Proposition \ref{prop:column-description},  a column in the
$\Z_2$ part of $A^{t_1,t_2,1}$ has the form $(1,\lambda,\rho)^T$, where
$\lambda=(\lambda_1,\dots,\lambda_{t_1})\in\Z_2^{t_1}$ and
$\rho=(\rho_1,\dots,\rho_{t_2})\in\Z_2^{t_2}$ are arbitrary.  Since $x_i\equiv a_i$ and
$y_j\equiv d_j \pmod 2$, its binary coordinate function is
$\epsilon+\sum_i\lambda_ia_i+\sum_j\rho_jd_j$.  Choosing $\lambda=\rho=0$ gives $\epsilon$. Taking $\rho=0$ and $\lambda$ to have exactly one nonzero entry gives $\epsilon+a_i$, for each $1\leq i\leq t_1$, while taking $\lambda=0$ and $\rho$ to have exactly one
nonzero entry gives $\epsilon+d_j$, for each $1\leq j\leq t_2$.  Adding
$\epsilon$ then gives each $a_i$ and each $d_j$. Hence, the coordinates in the $\Z_2$
part span exactly the space with basis $\epsilon$, $a_1,\dots,a_{t_1}$ and
$d_1,\dots,d_{t_2}$.

Now, let $(2,p_1,\dots,p_{t_1},s_1,\dots,s_{t_2})^T$ be the column of a
coordinate in the $\Z_4$ part, where $(p,s)$ is normalized primitive in
$\Z_4^{t_1+t_2}$ by item (ii) of
Proposition \ref{prop:column-description}.  Write the binary expansions
$p_i=\lambda_i+2\mu_i$ and $s_j=\rho_j+2\sigma_j$ in $\Z_4$.  The entry of the
additive codeword \eqref{eq:general-message} at this coordinate is
\begin{equation}\label{eq:z4-symbol}
2\epsilon+\sum_{i=1}^{t_1}p_ix_i+\sum_{j=1}^{t_2}s_jy_j \pmod 4.
\end{equation}
Since $p_i x_i\equiv \lambda_i a_i+2(\lambda_i b_i+\mu_i a_i)\pmod 4$ and $s_j y_j\equiv \rho_j d_j+2(\rho_j e_j+\sigma_j d_j)\pmod 4$, 
the bottom digit of \eqref{eq:z4-symbol} is
\begin{equation}\label{eq:z4-low-digit}
\ell=\sum_{i=1}^{t_1}\lambda_ia_i+\sum_{j=1}^{t_2}\rho_jd_j,
\end{equation}
and, adding the bits $\lambda_ia_i$ and $\rho_jd_j$ and taking into account
that the carry produced by that addition is the second elementary symmetric
function of those bits, the top digit is
\begin{equation}\label{eq:z4-high-digit}
h=\epsilon+\sum_{i=1}^{t_1}(\lambda_ib_i+\mu_ia_i)
   +\sum_{j=1}^{t_2}(\rho_je_j+\sigma_jd_j)
   +\sym_2(\lambda_1a_1,\dots,\lambda_{t_1}a_{t_1},
        \rho_1d_1,\dots,\rho_{t_2}d_{t_2}).
\end{equation}
By Lemma \ref{lem:gray-digits}, the two binary coordinates it produces span the
same space as $\ell$ and $h$.

Put $S=\{i:\lambda_i=1\}$, $T=\{i:\mu_i=1\}$, $R=\{j:\rho_j=1\}$ and
$R'=\{j:\sigma_j=1\}$, and let
$L_{S,R}=\{a_i:i\in S\}\cup\{d_j:j\in R\}\subseteq L$.  Then,
\eqref{eq:z4-low-digit} and \eqref{eq:z4-high-digit} read
$$
\ell=\sum_{i\in S}a_i+\sum_{j\in R}d_j,\qquad
h=\epsilon+\sum_{i\in S}b_i+\sum_{i\in T}a_i+\sum_{j\in R}e_j
   +\sum_{j\in R'}d_j+\sym_2(L_{S,R}),
$$
and the vector $(p,s)$ is primitive exactly when $(S,R)\not=(\emptyset,
\emptyset)$.  Both $\ell$ and $h$ clearly belong to the space spanned by
\eqref{eq:V0-basis}, and so do the functions obtained from the
coordinates in the $\Z_2$ part; hence $V_0$ is contained in that space.

Conversely, all the functions in \eqref{eq:V0-basis} occur.  The elements
$\epsilon$, $a_i$ and $d_j$ were already obtained from the
coordinates in the $\Z_2$ part.  Choosing the coordinate whose column is
normalized primitive with $S=\{i\}$, $R=\emptyset$ and all higher digits equal
to zero gives
$h=\epsilon+b_i$, hence $b_i$; note that normalization only forces $\mu_i=0$,
which is what we have chosen.  Symmetrically, $S=\emptyset$, $R=\{j\}$ gives
$h=\epsilon+e_j$, hence $e_j$.  Choosing the next two nonzero positions, say
$S=\{i,k\}$ and $R=\emptyset$, gives $h=\epsilon+b_i+b_k+a_ia_k$, and therefore
$a_ia_k$; and, in the same way, $S=\{i\}$, $R=\{j\}$ gives $a_id_j$, and
$S=\emptyset$, $R=\{j,l\}$ gives $d_jd_l$.  Therefore, every product of two
distinct elements of $L$ belongs to $V_0$.

Finally, all the functions listed in \eqref{eq:V0-basis} are distinct
square-free monomials in the independent variables
$\epsilon,a_i,b_i,d_j,e_j$, hence linearly independent, so they form a basis.
Counting them gives \eqref{eq:V0-dim}.
\end{proof}

\subsection{The three binary digits of a coordinate in the $\Z_8$ part}
\label{subsec:z8-digits}

By item (iii) of Proposition \ref{prop:column-description}, a coordinate in
the $\Z_8$ part of $A^{t_1,t_2,1}$ has column
\begin{equation}\label{eq:z8-column}
(4,q_1,\dots,q_{t_1},2r_1,\dots,2r_{t_2})^T,
\end{equation}
where $q=(q_1,\dots,q_{t_1})$ is normalized primitive in $\Z_8^{t_1}$ and
$r=(r_1,\dots,r_{t_2})\in\Z_4^{t_2}$ is arbitrary.  We write the binary
expansions
\begin{equation}\label{eq:column-digits}
q_i=\lambda_i+2\mu_i+4\nu_i\ \text{ in }\Z_8,
\qquad
r_j=\rho_j+2\sigma_j\ \text{ in }\Z_4.
\end{equation}
The entry of the additive codeword at a coordinate with
column \eqref{eq:z8-column} is
\begin{equation}\label{eq:z8-symbol}
4\epsilon+\sum_{i=1}^{t_1}q_ix_i+2\sum_{j=1}^{t_2}r_jy_j\pmod 8.
\end{equation}
We compute its three binary digits in two steps.  The first lemma expands
one product; it is where the first genuinely nonlinear carry appears.

\begin{lemma}\label{lem:product-digits}
With the notation \eqref{eq:message-digits} and \eqref{eq:column-digits}, the
bottom, middle and top binary digits of $q_ix_i$ modulo $8$ are,
respectively,
\begin{equation}\label{eq:ABC}
A_i=\lambda_ia_i,\qquad
B_i=\lambda_ib_i+\mu_ia_i,\qquad
C_i=\lambda_ic_i+\mu_ib_i+\nu_ia_i+\lambda_i\mu_ia_ib_i,
\end{equation}
and the middle and top binary digits of $2r_jy_j$ modulo $8$ are
\begin{equation}\label{eq:DE}
D_j=\rho_jd_j,\qquad E_j=\rho_je_j+\sigma_jd_j.
\end{equation}
\end{lemma}

\begin{proof}
Expanding the product as an ordinary integer and deleting the multiples of
$8$ gives $q_ix_i\equiv \lambda_ia_i+2(\lambda_ib_i+\mu_ia_i)
+4(\lambda_ic_i+\mu_ib_i+\nu_ia_i)\pmod 8$, where each of the six displayed
products of digits is a bit and the sums are
integer sums.  The bottom digit is the bit $\lambda_ia_i$, which gives
$A_i$.  The two bits $\lambda_ib_i$ and $\mu_ia_i$ occurring at the level of
$2$ add up to an integer in $\{0,1,2\}$, whose lowest bit is
$\lambda_ib_i+\mu_ia_i$ in $\Z_2$ and whose carry is
$(\lambda_ib_i)(\mu_ia_i)=\lambda_i\mu_ia_ib_i$; this proves the formula for
$B_i$ and shows that the carry $\lambda_i\mu_ia_ib_i$ is added at the level
of $4$.  At that level, the three bits $\lambda_ic_i$, $\mu_ib_i$, $\nu_ia_i$
and the carry are added, and only the parity of the result matters modulo
$8$; this gives $C_i$.

For the second statement, $r_jy_j\equiv\rho_jd_j+2(\rho_je_j+\sigma_jd_j)
\pmod 4$ by the same computation with one digit less, and multiplying by $2$
shifts these two digits to the levels of $2$ and $4$ of a symbol modulo $8$.
\end{proof}

The second lemma is the rule that governs the carries of an arbitrary sum of
bits.  It is classical, and we include the short proof for completeness.

\begin{lemma}\label{lem:carries}
Let $z_1,\dots,z_h\in\Z_2$ and let $w=z_1+\cdots+z_h$ be their ordinary
integer sum.  Then, for every $r\geq0$, the $r$th binary digit of $w$ equals
$\sym_{2^r}(z_1,\dots,z_h)$.  In particular, the bottom digit, the carry to the
next digit and the direct carry to the third digit are $\sym_1$, $\sym_2$, and
$\sym_4$, respectively.
\end{lemma}

\begin{proof}
Fix an assignment of the $z_i$. Then, $w$ is the Hamming weight of
$(z_1,\dots,z_h)$.  On this assignment exactly $\binom{w}{2^r}$ of the
products occurring in $\sym_{2^r}(z_1,\dots,z_h)$ equal $1$, so
$\sym_{2^r}(z_1,\dots,z_h)\equiv\binom{w}{2^r}\pmod 2$. Writing
$w=\sum_{i\geq0}w_i2^i$ and applying Lucas' theorem modulo $2$, and using the fact that
$2^r$ has a single nonzero binary digit, we get $\binom{w}{2^r}\equiv w_r$.
As $w_r$ is the $r$th binary digit of $w$ and the assignment was arbitrary,
the two Boolean functions coincide.
\end{proof}

\begin{lemma}\label{lem:z8-digits}
Fix a coordinate in the $\Z_8$ part of $A^{t_1,t_2,1}$, with column
\eqref{eq:z8-column}, and
let $A_i,B_i,C_i$ $(1\leq i\leq t_1)$ and $D_j,E_j$ $(1\leq j\leq t_2)$ be the
Boolean functions \eqref{eq:ABC} and \eqref{eq:DE} attached to it by Lemma
\ref{lem:product-digits}.  Then, the bottom, the middle and the top
binary digits of the entry \eqref{eq:z8-symbol} are, respectively,
\begin{align}
P_{q,r}&=\sym_1(A_1,\dots,A_{t_1}),\label{eq:bottom-digit}\\
Q_{q,r}&=\sym_1(B_1,\dots,B_{t_1},D_1,\dots,D_{t_2})
        +\sym_2(A_1,\dots,A_{t_1}),\label{eq:middle-digit}\\
F_{q,r}&=\epsilon+\sum_{i=1}^{t_1}C_i+\sum_{j=1}^{t_2}E_j
       +\sym_2(B_1,\dots,B_{t_1},D_1,\dots,D_{t_2})
       +\sym_4(A_1,\dots,A_{t_1})\notag\\
&\qquad
       +\sym_1(B_1,\dots,B_{t_1},D_1,\dots,D_{t_2})\,
        \sym_2(A_1,\dots,A_{t_1}).\label{eq:top-digit}
\end{align}
Moreover, $P_{q,r}\in V_0$ and $Q_{q,r}\in V_0$. The space
spanned by all the binary coordinate functions of $H^{t_1,t_2,1}$
is therefore
\begin{equation}\label{eq:total-span}
V_0+W,\qquad\text{where }\
W=\Span\bigl\{F_{q,r}:\ q\ \text{normalized primitive in }\Z_8^{t_1},\
r\in\Z_4^{t_2}\bigr\},
\end{equation}
and hence $\rank(H^{t_1,t_2,1})=\dim(V_0+W)$.
\end{lemma}

\begin{proof}
By Lemma \ref{lem:product-digits}, the symbol \eqref{eq:z8-symbol} equals
$\alpha+2\beta+4\gamma$ modulo $8$, where
$\alpha=\sum_iA_i$, $\beta=\sum_iB_i+\sum_jD_j$ and
$\gamma=\epsilon+\sum_iC_i+\sum_jE_j$ are integer sums of bits.  By Lemma
\ref{lem:carries}, the binary digits of $\alpha$ are $\sym_1(A)$,
$\sym_2(A)$ and $\sym_4(A)$, those of $\beta$ are $\sym_1(B,D)$ and
$\sym_2(B,D)$, and the bottom digit of $\gamma$ is
$\epsilon+\sum_iC_i+\sum_jE_j$ read in $\Z_2$.

The bottom digit of the symbol is therefore $\sym_1(A)$, which is
\eqref{eq:bottom-digit}.  At the level of $2$, the digits $\sym_2(A)$ of $\alpha$
and $\sym_1(B,D)$ of $\beta$ are added, so the middle digit is
$\sym_2(A)+\sym_1(B,D)$, which is \eqref{eq:middle-digit}, and the carry produced
by that addition is $\sym_1(B,D)\sym_2(A)$.  At the level of $4$, we add the digit
$\sym_4(A)$ of $\alpha$, the digit $\sym_2(B,D)$ of $\beta$, the digit
$\epsilon+\sum_iC_i+\sum_jE_j$ of $\gamma$ and the carry
$\sym_1(B,D)\sym_2(A)$, and only the parity matters modulo $8$; this is
\eqref{eq:top-digit}.

Next, $P_{q,r}=\sum_i\lambda_ia_i$ is a linear combination of the $a_i$, hence
lies in $V_0$, and $Q_{q,r}=\sum_i(\lambda_ib_i+\mu_ia_i)+\sum_j\rho_jd_j
+\sum_{i<k}\lambda_i\lambda_ka_ia_k$ is a linear combination of the functions
$b_i$, $a_i$, $d_j$ and $a_ia_k$, all
of which belong to the basis \eqref{eq:V0-basis} of $V_0$.

For the last assertion, recall from Lemma \ref{lem:gray-digits} that the four
binary coordinates produced by a coordinate with column \eqref{eq:z8-column}
span the same space of Boolean functions as its three digits $P_{q,r}$,
$Q_{q,r}$ and $F_{q,r}$.  Adding these spans over all the coordinates in the $\Z_8$ part,
and adding the space $V_0$ produced by those in the $\Z_2$ and the $\Z_4$
parts, we conclude that all the binary coordinate functions of $H^{t_1,t_2,1}$
span $V_0+W$ with $W$ as in \eqref{eq:total-span}, because $P_{q,r}$ and
$Q_{q,r}$ lie in $V_0$ and the coordinates in the $\Z_8$ part therefore contribute
nothing beyond $V_0$ except the functions $F_{q,r}$.  By
\eqref{eq:rank-as-function-span}, this gives
$\rank(H^{t_1,t_2,1})=\dim(V_0+W)$.
\end{proof}

\subsection{The reduced top digit}\label{subsec:reduced-top}

The expression \eqref{eq:top-digit} looks complicated, but almost all of it
disappears modulo $V_0$.  The next lemma performs that reduction and is the
technical heart of the section. It shows that, modulo $V_0$, the top digit
depends on the coordinate through two data only, namely the set $S$ of positions
where $q$ is odd and one linear form in the variables of $L$.  For a subset
$S\subseteq\{1,\dots,t_1\}$ and a subset $R\subseteq\{1,\dots,t_2\}$, we
write
\begin{equation}\label{eq:symmetric-notation}
\begin{aligned}
&\alpha_S=\sum_{i\in S}a_i,\quad
\beta_S=\sum_{i\in S}b_i,\quad
\gamma_S=\sum_{i\in S}c_i,\quad
\delta_R=\sum_{j\in R}d_j,\\
&\sym_2(a_S)=\sum_{\{i,k\}\subseteq S}a_ia_k,\quad
\sym_2(b_S)=\sum_{\{i,k\}\subseteq S}b_ib_k,\quad
\sym_4(a_S)=\sum_{\{i,j,k,l\}\subseteq S}a_ia_ja_ka_l,
\end{aligned}
\end{equation}
the last three sums being taken over all the subsets of $S$ of size $2$, $2$
and $4$, respectively.

\begin{lemma}\label{lem:top-digit-reduced}
Let $(4,q_1,\dots,q_{t_1},2r_1,\dots,2r_{t_2})^T$ be the column of a coordinate
in the $\Z_8$ part of $A^{t_1,t_2,1}$, let $S=\{i:q_i\text{ is odd}\}$ and
$p=\min S$, and let the
digits of $q_i$ and $r_j$ be as in \eqref{eq:column-digits}.  Put
$T=\{i:\mu_i=1\}$ and $R=\{j:\rho_j=1\}$.  Then, $S\not=\emptyset$,
$p\notin T$, and, modulo $V_0$,
\begin{equation}\label{eq:top-digit-reduced}
F_{q,r}\equiv g_S+\ell\cdot\bigl(\beta_S+\sym_2(a_S)\bigr),
\end{equation}
where
\begin{equation}\label{eq:gS}
g_S=\gamma_S+\sym_2(b_S)+\sym_4(a_S)+\beta_S\,\sym_2(a_S)
\qquad\text{and}\qquad
\ell=\alpha_T+\delta_R.
\end{equation}
Conversely, for every nonempty $S\subseteq\{1,\dots,t_1\}$ and every linear
form $\ell$ in the variables of $L\setminus\{a_{\min S}\}$, there is a
coordinate in the $\Z_8$ part of $A^{t_1,t_2,1}$ realising the pair $(S,\ell)$.
\end{lemma}

\begin{proof}
Since $q$ is primitive, $S\not=\emptyset$; and since $q$ is normalized, its
first odd entry satisfies $q_p=1$, that is $\lambda_p=1$ and
$\mu_p=\nu_p=0$, and therefore $p\notin T$.  All the remaining digits
$\mu_i,\nu_i$ with $i\not=p$ and all the digits $\rho_j,\sigma_j$ are free,
because item (iii) of Proposition \ref{prop:column-description} imposes no
further condition. The entries $q_i$ with $i<p$ are even, which is exactly
the condition $\lambda_i=0$ for $i<p$, and their higher digits are free.
Consequently, as the coordinate varies, the pair $(T,R)$ runs over all the
pairs with $T\subseteq\{1,\dots,t_1\}\setminus\{p\}$ and
$R\subseteq\{1,\dots,t_2\}$, so $\ell=\alpha_T+\delta_R$ runs over all the
linear forms in the variables of $L\setminus\{a_p\}$.  This proves the last
assertion, and it only remains to establish \eqref{eq:top-digit-reduced}.

We reduce the five summands of \eqref{eq:top-digit} one by one, using the
basis \eqref{eq:V0-basis} of $V_0$ and writing $U=\{i:\nu_i=1\}$ and
$R'=\{j:\sigma_j=1\}$.  Note first that $A_i=\lambda_ia_i$ vanishes unless
$i\in S$, so $\sym_2(A)=\sym_2(a_S)$ and $\sym_4(A)=\sym_4(a_S)$; and that
$B_i=\lambda_ib_i+\mu_ia_i$ and $D_j=\rho_jd_j$, so 
$\sym_1(B,D)=\beta_S+\alpha_T+\delta_R$.

The term $\epsilon$ belongs to $V_0$.  Next, by \eqref{eq:ABC},
$$
\sum_{i}C_i=\sum_{i\in S}c_i+\sum_{i\in T}b_i+\sum_{i\in U}a_i
            +\sum_{i\in S\cap T}a_ib_i
\equiv\gamma_S+\sum_{i\in S\cap T}a_ib_i \pmod{V_0},
$$
because $b_i$ and $a_i$ lie in $V_0$, whereas $a_ib_i$ does not, since $b_i$
is not an element of $L$.  Similarly, by \eqref{eq:DE},
$\sum_jE_j=\sum_{j\in R}e_j+\sum_{j\in R'}d_j\in V_0$.

We now expand $\sym_2(B,D)$.  First,
$$
\sum_{i<k}B_iB_k=\sum_{i<k}(\lambda_ib_i+\mu_ia_i)(\lambda_kb_k+\mu_ka_k)
=\sym_2(b_S)+\sum_{i\not=k}\lambda_i\mu_kb_ia_k
 +\sum_{i<k}\mu_i\mu_ka_ia_k,
$$
whose last sum lies in $V_0$, being a combination of products of two distinct
elements of $L$, while its middle sum equals $\sum_{i\in S}\sum_{k\in T,\,
k\not=i}a_kb_i$ because $\lambda_i=1$ exactly for $i\in S$ and $\mu_k=1$
exactly for $k\in T$. Second,
$$\sum_{i,j}B_iD_j=\sum_{i\in S,\,j\in R}b_id_j
+\sum_{i\in T,\,j\in R}a_id_j\equiv\beta_S\delta_R\pmod{V_0},$$ again because
$a_id_j\in V_0$. Third, $\sum_{j<l}D_jD_l=\sum_{\{j,l\}\subseteq R}d_jd_l\in V_0$.
Adding the three contributions,
$$
\sym_2(B,D)\equiv \sym_2(b_S)+\sum_{i\in S}\sum_{k\in T,\,k\not=i}a_kb_i
             +\beta_S\delta_R\pmod{V_0}.
$$
Combining this with the term $\sum_{i\in S\cap T}a_ib_i$ obtained above and
observing that
$\sum_{i\in S\cap T}a_ib_i+\sum_{i\in S}\sum_{k\in T,\,k\not=i}a_kb_i
=\sum_{i\in S}\sum_{k\in T}a_kb_i=\beta_S\alpha_T,$
we get
$$
\sum_iC_i+\sum_jE_j+\sym_2(B,D)\equiv
\gamma_S+\sym_2(b_S)+\beta_S\alpha_T+\beta_S\delta_R\pmod{V_0}.
$$
Finally, $\sym_4(A)=\sym_4(a_S)$ and
$\sym_1(B,D)\sym_2(A)=(\beta_S+\alpha_T+\delta_R)\sym_2(a_S)$. Adding
everything and writing $\ell=\alpha_T+\delta_R$, we obtain
$F_{q,r}\equiv\gamma_S+\sym_2(b_S)+\sym_4(a_S)+\beta_S\,\sym_2(a_S)
        +\ell\bigl(\beta_S+\sym_2(a_S)\bigr)\pmod{V_0},$
which is \eqref{eq:top-digit-reduced}.
\end{proof}

Formula \eqref{eq:top-digit-reduced} is affine in $\ell$, which is what makes
the dimension count of the next subsection short: for each fixed $S$, we only
have to add the single function $g_S$, obtained for $\ell=0$, and the products
$v\cdot(\beta_S+\sym_2(a_S))$ with $v$ ranging over the variables of
$L\setminus\{a_{\min S}\}$.

\begin{example}\label{ex:reduced-t1-1}
Let $t_1=1$.  Then $S=\{1\}$ is forced, $p=1$, $\beta_S=b_1$, $\sym_2(a_S)=0$,
$\sym_2(b_S)=0$, $\sym_4(a_S)=0$ and $g_S=c_1$.  Moreover, $T\subseteq\emptyset$,
so $\ell=\delta_R$ is a linear form in the $d_j$ only, and
\eqref{eq:top-digit-reduced} becomes $F_{q,r}\equiv c_1+(\sum_{j\in R}d_j)b_1\pmod{V_0}$. Hence, modulo $V_0$, the top digits span $\Span\{c_1,b_1d_1,\dots,
b_1d_{t_2}\}$, a space of dimension $1+t_2$.
\end{example}

\subsection{The additional carry space}\label{subsec:carry-space}

We now compute the dimension of $(V_0+W)/V_0$, where $W$ is as in
\eqref{eq:total-span}.  By Lemma \ref{lem:top-digit-reduced} and the fact
that \eqref{eq:top-digit-reduced} is affine in $\ell$, we have
\begin{equation}\label{eq:W-decomposition}
\frac{V_0+W}{V_0}=\overline{W_1}+\overline{W_2},
\end{equation}
where the bars denote images in the quotient by $V_0$ and
\begin{equation}\label{eq:W1W2}
\begin{aligned}
W_1&=\Span\bigl\{v\cdot(\beta_S+\sym_2(a_S)):\
 \emptyset\not=S\subseteq\{1,\dots,t_1\},\
 v\in L\setminus\{a_{\min S}\}\bigr\},\\
W_2&=\Span\bigl\{g_S:\ \emptyset\not=S\subseteq\{1,\dots,t_1\}\bigr\}.
\end{aligned}
\end{equation}
Indeed, for a fixed $S$, the set of functions
$\{g_S+\ell(\beta_S+\sym_2(a_S))\}$, with $\ell$ running over all the linear
forms in $L\setminus\{a_{\min S}\}$, spans
$\Span\{g_S\}+\Span\{v(\beta_S+\sym_2(a_S)):v\in L\setminus\{a_{\min S}\}\}$,
because the map $\ell\mapsto \ell(\beta_S+\sym_2(a_S))$ is $\Z_2$-linear in
$\ell$ and $\ell=0$ is allowed.

We compute the two pieces separately and then show that they meet only in
$V_0$.  Throughout, ``modulo $V_0$'' means that the monomials occurring in
\eqref{eq:V0-basis}, namely $\epsilon$, the variables $a_i$, $b_i$, $d_j$,
$e_j$, and the products of two distinct elements of $L$, are deleted.

\begin{lemma}\label{lem:W1}
The space $\overline{W_1}$ has as a basis the set of monomials
\begin{equation}\label{eq:W1-basis}
\begin{gathered}
a_kb_i\ (1\leq i,k\leq t_1,\ i\not=k),\qquad
a_kb_k\ (2\leq k\leq t_1),\\
a_ia_ja_k\ (1\leq i<j<k\leq t_1),\qquad
b_id_j\ (1\leq i\leq t_1,\ 1\leq j\leq t_2),\\
a_ia_kd_j\ (1\leq i<k\leq t_1,\ 1\leq j\leq t_2),
\end{gathered}
\end{equation}
and therefore
\begin{equation}\label{eq:W1-dim}
\dim\overline{W_1}=t_1(t_1-1)+(t_1-1)+\binom{t_1}{3}
 +t_2\binom{t_1+1}{2}.
\end{equation}
\end{lemma}

\begin{proof}
Write $p=\min S$ and set $h_{S,v}=v\,(\beta_S+\sym_2(a_S))$ for
$v\in L\setminus\{a_p\}$.  There are two kinds of generators.

Assume first that $v=d_j$.  Then
\begin{equation}\label{eq:hSdj}
h_{S,d_j}=\sum_{i\in S}b_id_j+\sum_{\{i,k\}\subseteq S}a_ia_kd_j,
\end{equation}
and no reduction modulo $V_0$ is needed, because all the displayed monomials
have degree at least two and involve the variable $d_j$ together with a variable $b_i$ or with two distinct variables $a_i,a_k$, and none of these monomials occurs in \eqref{eq:V0-basis}. Taking $S=\{i\}$ in \eqref{eq:hSdj} gives $b_id_j$, and then taking $S=\{i,k\}$ and subtracting $b_id_j$ and $b_kd_j$ gives $a_ia_kd_j$.  Conversely, every $h_{S,d_j}$ is by \eqref{eq:hSdj} a sum of
monomials of these two shapes.  Note that the constraint $v\not=a_p$ never
restricts the choice of $v=d_j$.

Assume now that $v=a_k$ with $k\not=p$.  Then, $h_{S,a_k}=\sum_{i\in
S}a_kb_i+a_k\,\sym_2(a_S)$, and $a_k\,\sym_2(a_S)=\sum_{\{i,l\}\subseteq
S}a_ka_ia_l$.  If $k\in S$, the
pairs $\{i,l\}\subseteq S$ containing $k$ contribute the products
$a_ka_l$ with $l\in S\setminus\{k\}$, which are products of two distinct
elements of $L$ and therefore lie in $V_0$; the pairs not containing $k$
contribute monomials $a_ia_la_k$ with three distinct indices.  If $k\notin
S$, all the pairs contribute monomials with three distinct indices.  Hence,
in both cases,
\begin{equation}\label{eq:hSak}
h_{S,a_k}\equiv\sum_{i\in S}a_kb_i
 +\sum_{\{i,l\}\subseteq S\setminus\{k\}}a_ia_la_k \pmod{V_0}.
\end{equation}
Every monomial occurring in \eqref{eq:hSak} is of the form $a_kb_i$ with
$i\not=k$, or of the form $a_kb_k$ with $k\in S$ and $k\not=p=\min S$, which
forces $k\geq2$, or of the form $a_ia_la_k$ with three distinct indices.  All
of them are listed in \eqref{eq:W1-basis}.

It remains to show that all the monomials of \eqref{eq:W1-basis} occur.
Taking $S=\{i\}$ with $i\not=k$ in \eqref{eq:hSak}, the second sum is empty
and we obtain $a_kb_i$; here the constraint $k\not=p=i$ is exactly $i\not=k$,
so all such monomials are obtained.  Taking $S=\{i,l\}$ and $k\notin S$, which
automatically satisfies $k\not=\min S$, the first sum reduces modulo the
monomials already obtained and leaves $a_ia_la_k$; every three-element subset
$\{\alpha<\beta<\gamma\}$ is reached by choosing $S=\{\alpha,\beta\}$ and
$k=\gamma$.  Finally, taking $S=\{i,k\}$ with $i<k$ and $v=a_k$, which is
legitimate because $k\not=\min S=i$, formula \eqref{eq:hSak} gives
$a_kb_i+a_kb_k+a_ia_k$, and deleting $a_ia_k\in V_0$ and the monomial $a_kb_i$
already obtained leaves $a_kb_k$; every $k$ with $2\leq k\leq t_1$ arises in
this way.

Therefore, $\overline{W_1}$ is spanned by the monomials \eqref{eq:W1-basis} and
contains all of them.  Being pairwise distinct square-free monomials, they are
linearly independent and form a basis. Counting them gives
$
t_1(t_1-1)+(t_1-1)+\binom{t_1}{3}+t_1t_2+t_2\binom{t_1}{2}
=t_1(t_1-1)+(t_1-1)+\binom{t_1}{3}+t_2\binom{t_1+1}{2},
$
where we used $t_1+\binom{t_1}{2}=\binom{t_1+1}{2}$.
\end{proof}

\begin{lemma}\label{lem:W2}
The space $\overline{W_2}$ has dimension
\begin{equation}\label{eq:W2-dim}
\dim\overline{W_2}=t_1+\binom{t_1}{2}+\binom{t_1}{3}+\binom{t_1}{4}.
\end{equation}
\end{lemma}

\begin{proof}
Recall from \eqref{eq:gS} that
$g_S=\gamma_S+G_S$ with
$G_S=\sym_2(b_S)+\sym_4(a_S)+\beta_S\,\sym_2(a_S)$.  If $S=\{i\}$ is a singleton, then
$\sym_2(b_S)=\sym_2(a_S)=\sym_4(a_S)=0$, so $g_{\{i\}}=c_i$.  Since
$\gamma_S=\sum_{i\in S}c_i$ always belongs to $\Span\{c_1,\dots,c_{t_1}\}$,
we obtain $W_2=\Span\{c_1,\dots,c_{t_1}\}+\Span\{G_S:\
S\subseteq\{1,\dots,t_1\},\ |S|\geq2\}$. The variables $c_i$ occur in no
$G_S$, so the two summands intersect trivially
and $\dim\overline{W_2}=t_1+\dim\Span\{G_S:|S|\geq2\}$. Note that none of the
monomials involved lies in $V_0$, since they all have degree at least two and
none of them is a product of two distinct elements of $L$.

We now compute $\dim\Span\{G_S:|S|\geq2\}$.  For a two-element subset
$P=\{i,k\}$, a three-element subset $T=\{i,k,l\}$ and a four-element subset
$F=\{i,j,k,l\}$ of $\{1,\dots,t_1\}$, put
$u_P=b_ib_k+a_ia_kb_i+a_ia_kb_k$,
$v_T=a_ia_kb_l+a_ia_lb_k+a_ka_lb_i$,
$w_F=a_ia_ja_ka_l$.
Expanding $\beta_S\,\sym_2(a_S)=\sum_{i\in S}\sum_{\{k,l\}\subseteq S}b_ia_ka_l$
and separating the terms with $i\in\{k,l\}$ from those with $i\notin\{k,l\}$,
we obtain
\begin{equation}\label{eq:GS-expansion}
G_S=\sum_{\substack{P\subseteq S\\|P|=2}}u_P
   +\sum_{\substack{T\subseteq S\\|T|=3}}v_T
   +\sum_{\substack{F\subseteq S\\|F|=4}}w_F.
\end{equation}
The functions $u_P$, $v_T$ and $w_F$ are linearly independent: the monomial
$b_ib_k$ occurs only in $u_{\{i,k\}}$, a monomial $a_ia_kb_l$ with three
distinct indices occurs only in $v_{\{i,k,l\}}$, and a product of four
distinct variables $a$ occurs only in the corresponding $w_F$.

Identify a subset $S\subseteq\{1,\dots,t_1\}$ with its indicator vector
$\varsigma\in\Z_2^{t_1}$.  By \eqref{eq:GS-expansion}, the coefficient of
$u_P$, of $v_T$ or of $w_F$ in $G_S$ is $\prod_{i\in P}\varsigma_i$,
$\prod_{i\in T}\varsigma_i$ or $\prod_{i\in F}\varsigma_i$, respectively.
Hence, $\dim\Span\{G_S:|S|\geq2\}$ is the rank over $\Z_2$ of the matrix whose
rows are indexed by the vectors $\varsigma$ of weight at least two and whose
columns are indexed by the subsets $I$ with $2\leq|I|\leq4$, the entry being
$\prod_{i\in I}\varsigma_i$.  Each column is the evaluation vector of a
distinct square-free monomial of
degree $2$, $3$ or $4$ in $\varsigma_1,\dots,\varsigma_{t_1}$.  Such monomials
are linearly independent as functions on $\Z_2^{t_1}$, and all of them vanish
at every $\varsigma$ of weight at most one. So, a linear combination of them
vanishing on all the vectors of weight at least two vanishes identically and
is the zero combination.  Hence, the columns are linearly independent, the rank
is their number $\binom{t_1}{2}+\binom{t_1}{3}+\binom{t_1}{4}$, and adding the
$t_1$ dimensions contributed by the $c_i$ gives \eqref{eq:W2-dim}.
\end{proof}

\begin{proposition}\label{prop:carry-dimension}
Let $t_1\geq1$ and $t_2\geq0$.  Then,
\begin{equation}\label{eq:carry-dimension}
\dim\frac{V_0+W}{V_0}
=-1+2t_1+3\binom{t_1}{2}+2\binom{t_1}{3}+\binom{t_1}{4}
 +t_2\binom{t_1+1}{2}.
\end{equation}
\end{proposition}

\begin{proof}
By Lemma \ref{lem:W1}, the monomials occurring in $\overline{W_1}$ have the
shapes $a_kb_i$, $a_ia_ja_k$, $b_id_j$ and $a_ia_kd_j$. By the proof of Lemma
\ref{lem:W2}, those occurring in $\overline{W_2}$ have the shapes $c_i$,
$b_ib_k$, $a_ia_kb_l$ and $a_ia_ja_ka_l$.  The two lists are disjoint, since a
monomial of the first has degree two with one variable $a$ and one variable
$b$, or degree three with three variables $a$, or contains a variable $d$,
whereas one of the second contains a variable $c$, or has degree two with two
variables $b$, or degree three with two variables $a$ and one variable $b$, or
degree four with four variables $a$.  Consequently,
$\overline{W_1}\cap\overline{W_2}=\{0\}$ and, by
\eqref{eq:W-decomposition},
$\dim\bigl((V_0+W)/V_0\bigr)=\dim\overline{W_1}+\dim\overline{W_2}$. Adding
\eqref{eq:W1-dim} and \eqref{eq:W2-dim}, we obtain
$$
t_1(t_1-1)+(t_1-1)+\binom{t_1}{3}+t_2\binom{t_1+1}{2}
+t_1+\binom{t_1}{2}+\binom{t_1}{3}+\binom{t_1}{4},
$$
and, since $t_1(t_1-1)=2\binom{t_1}{2}$ and
$t_1(t_1-1)+t_1-1+t_1=t_1^2+t_1-1=2t_1+2\binom{t_1}{2}-1$, this equals the
right-hand side of \eqref{eq:carry-dimension}.
\end{proof}

\begin{example}\label{ex:carry-t1-2}
Let $t_1=2$ and $t_2=0$.  Formula \eqref{eq:carry-dimension} gives
$-1+4+3=6$.  Explicitly, the basis of $\overline{W_1}$ given by
\eqref{eq:W1-basis} is $\{a_2b_1,a_1b_2,a_2b_2\}$, of size
$t_1(t_1-1)+(t_1-1)=2+1=3$, while $\overline{W_2}$ has basis
$\{c_1,c_2,G_{\{1,2\}}\}$ with $G_{\{1,2\}}=b_1b_2+a_1a_2b_1+a_1a_2b_2$, again
of size $3$.  Note that $a_1b_1$ does not occur: by Lemma
\ref{lem:W1}, a monomial $a_kb_k$ requires $k\in S$ and $k\not=\min S$, and
$k=1$ can never satisfy both conditions.  This is precisely why the count in
\eqref{eq:W1-dim} contains $t_1-1$ and not $t_1$ such monomials.
\end{example}

\subsection{The rank formula}\label{subsec:rank-formula}

It only remains to add the rows of order $2$.  We recall the following result
from \cite[Theorem~10]{Z2Z4Z8Linearity}, which will also be used in Section
\ref{sec:rk}.

\begin{lemma}\cite{Z2Z4Z8Linearity}\label{lem:plotkin}
Let $t_1\geq1$, $t_2\geq0$, and $t_3\geq1$.  Then,
\begin{equation}\label{eq:plotkin-code}
\rank\bigl(H^{t_1,t_2,t_3+1}\bigr)=1+\rank\bigl(H^{t_1,t_2,t_3}\bigr).
\end{equation}
\end{lemma}

\begin{theorem}\label{thm:rank-main}
Let $t_1\geq1$, $t_2\geq0$ and $t_3\geq1$ be integers.  Then,
\begin{equation}\label{eq:rank-binomial}
\rank(H^{t_1,t_2,t_3})=t_3-1+4t_1+4\binom{t_1}{2}+2\binom{t_1}{3}
+\binom{t_1}{4}+t_2\binom{t_1+2}{2}+\binom{t_2+1}{2}.
\end{equation}
Equivalently,
\begin{equation}\label{eq:rank-polynomial}
\rank(H^{t_1,t_2,t_3})=\frac{t_1^4+2t_1^3+35t_1^2+58t_1}{24}
+\frac{t_2}{2}\bigl(t_1^2+3t_1+t_2+3\bigr)+t_3-1,
\end{equation}
and, if the binary length is $2^t$, then
\begin{equation}\label{eq:rank-length-form}
\rank(H^{t_1,t_2,t_3})=t+\frac{t_1^4+2t_1^3+35t_1^2-14t_1}{24}
+\frac{t_2}{2}\bigl(t_1^2+3t_1+t_2-1\bigr).
\end{equation}
\end{theorem}

\begin{proof}
Assume first that $t_3=1$.  By \eqref{eq:rank-as-function-span} and
\eqref{eq:total-span}, the rank of $H^{t_1,t_2,1}$ is the dimension of
$V_0+W$.  By Lemma \ref{lem:V0} and Proposition \ref{prop:carry-dimension},
\begin{align*}
\rank(H^{t_1,t_2,1})
&=\dim V_0+\dim\frac{V_0+W}{V_0}\\
&=1+2(t_1+t_2)+\binom{t_1+t_2}{2}
 -1+2t_1+3\binom{t_1}{2}+2\binom{t_1}{3}+\binom{t_1}{4}
 +t_2\binom{t_1+1}{2}.
\end{align*}
Using $\binom{t_1+t_2}{2}=\binom{t_1}{2}+t_1t_2+\binom{t_2}{2}$,
$\binom{t_1+1}{2}=t_1+\binom{t_1}{2}$,
$\binom{t_1+2}{2}=\binom{t_1}{2}+2t_1+1$ and
$\binom{t_2+1}{2}=\binom{t_2}{2}+t_2$, the right-hand side becomes
$$
4t_1+4\binom{t_1}{2}+2\binom{t_1}{3}+\binom{t_1}{4}
+t_2\binom{t_1+2}{2}+\binom{t_2+1}{2},
$$
which is \eqref{eq:rank-binomial} for $t_3=1$.

For $t_3>1$, applying \eqref{eq:plotkin-code} exactly $t_3-1$ times gives
$\rank(H^{t_1,t_2,t_3})=\rank(H^{t_1,t_2,1})+(t_3-1)$, which is exactly the
term $t_3-1$ in \eqref{eq:rank-binomial}.

Expanding the binomial coefficients,
$4t_1+4\binom{t_1}{2}+2\binom{t_1}{3}+\binom{t_1}{4}
=(t_1^4+2t_1^3+35t_1^2+58t_1)/24$, and
$t_2\binom{t_1+2}{2}+\binom{t_2+1}{2}
=\frac{t_2}{2}(t_1^2+3t_1+2)+\frac{t_2}{2}(t_2+1)
=\frac{t_2}{2}(t_1^2+3t_1+t_2+3)$, which proves
\eqref{eq:rank-polynomial}.  Finally, substituting $t_3=t+1-3t_1-2t_2$ from
\eqref{eq:length-relation} into \eqref{eq:rank-polynomial} and simplifying
gives \eqref{eq:rank-length-form}.
\end{proof}

\begin{example}\label{ex:rank-H111}
Consider $H^{1,1,1}$, of length $2^5=32$, with generator matrix
\eqref{eq:A111-reordered}.  Writing $x=a+2b+4c$ in $\Z_8$, $y=d+2e$ in $\Z_4$
and $\epsilon$ for the coefficient of the distinguished row, Lemma
\ref{lem:V0} gives the basis $\{\epsilon,a,b,d,e,ad\}$ of $V_0$, so
$\dim V_0=6$; and by Example \ref{ex:reduced-t1-1}, the top digits contribute,
modulo $V_0$, the two functions $c$ and $bd$.  Hence,
$\rank(H^{1,1,1})=6+2=8$, in agreement with \eqref{eq:rank-binomial}, which
gives $4+\binom32+\binom22=8$.  Similarly, for $H^{2,0,1}$ of length $2^6$,
Lemma \ref{lem:V0} gives $\dim V_0=6$ and Example \ref{ex:carry-t1-2} gives a carry
space of dimension $6$, so $\rank(H^{2,0,1})=12$, while
\eqref{eq:rank-binomial} gives $8+4=12$.  Two larger values are
$\rank(H^{3,1,1})=37$ and $\rank(H^{5,0,1})=85$, of lengths $2^{11}$ and
$2^{15}$.  All these values agree with the ones computed with {\sc Magma} in
\cite[Tables 2--4]{Z2Z4Z8Linearity}.
\end{example}

The next corollary lists the values used repeatedly in Section \ref{sec:rk}.
Each is obtained by substituting the displayed parameters into
\eqref{eq:rank-binomial}, and we omit the routine verification.

\begin{corollary}\label{cor:rank-special-values}
For every value of $t$ for which the third parameter is positive,
$$
\begin{array}{lll}
\rank(H^{1,0,t-2})=t+1, &
\rank(H^{1,1,t-4})=t+3, &
\rank(H^{1,2,t-6})=t+6,\\[1mm]
\rank(H^{2,0,t-5})=t+6, &
\rank(H^{1,3,t-8})=t+10, &
\rank(H^{2,1,t-7})=t+11,\\[1mm]
\rank(H^{1,4,t-10})=t+15, &
\rank(H^{2,2,t-9})=t+17, &
\rank(H^{3,0,t-8})=t+17,\\[1mm]
\rank(H^{1,5,t-12})=t+21, &
\rank(H^{3,1,t-10})=t+26, &
\rank(H^{2,6,t-17})=t+51,\\[1mm]
\rank(H^{4,5,t-21})=t+117. & &
\end{array}
$$
\end{corollary}

\begin{remark}\label{rem:rank-consistency}
Formula \eqref{eq:rank-binomial} passes two immediate consistency checks.
For $(t_1,t_2)=(1,0)$, it gives
$\rank(H^{1,0,t_3})=t_3+3$, and, by \eqref{eq:length-relation}, the length of
this code is $2^t$ with $t=t_3+2$; hence
$\rank(H^{1,0,t_3})=t+1$, which is exactly the dimension of the binary linear
Hadamard code of length $2^t$, in agreement with item (i) of Theorem
\ref{thm:recalled-linearity-kernel}.  This also explains why item (ii) of
that theorem is stated only for the nonlinear members: in the linear case
the dimension of the kernel is $t+1$ and not $t_1+t_2+t_3$.  For every
nonlinear member, on the other hand, a direct comparison of
\eqref{eq:rank-binomial} with
$\kernel(H^{t_1,t_2,t_3})=t_1+t_2+t_3$ shows that
$\rank(H^{t_1,t_2,t_3})>\kernel(H^{t_1,t_2,t_3})$, as it must be for a
nonlinear code.
\end{remark}
\section{Classification by the rank and the dimension of the kernel}
\label{sec:rk}

Now that the rank is known in closed form, we can determine exactly how far
the rank and the dimension of the kernel classify the
$\Z_2\Z_4\Z_8$-linear Hadamard codes. By item (i) of Theorem
\ref{thm:recalled-linearity-kernel}, the only linear member of length $2^t$ is
$H^{1,0,t-2}$, and linearity is an equivalence invariant, so it never has to
be compared with a nonlinear one.  By item (ii), a nonlinear member has
$\kernel(H^{t_1,t_2,t_3})=t_1+t_2+t_3$, and by \eqref{eq:length-relation} its
length is $2^t$ with $t+1=3t_1+2t_2+t_3$.  Hence, its length and kernel
dimension determine
\begin{equation}\label{eq:u-definition}
\kappa=t_1+t_2+t_3
\qquad\text{and}\qquad
u=t+1-\kappa=2t_1+t_2,
\end{equation}
and conversely $t_2=u-2t_1$ and $t_3=\kappa-t_1-t_2$.  In other words, once
the length and the dimension of the kernel are fixed, the whole triple is
determined by the single parameter $t_1$.  Two nonlinear members of the family
of the same length and with the same dimension of the kernel are therefore
distinguished by $t_1$ alone, and the question is whether the rank sees $t_1$.

Subsection \ref{subsec:comparison-families} recalls the two comparison
families with their length, rank and kernel.  Subsection \ref{subsec:reduced-rank} introduces the \emph{reduced rank}, a normalised
form of the rank obtained by subtracting the part common to all Hadamard codes
of the same length and kernel dimension. This reduces each comparison below into the
comparison of two small polynomial expressions. It also reveals a fact that is not apparent from the original formulas: the ranks of the three families are values of one and the same polynomial.  Subsections \ref{subsec:internal-rk} and \ref{subsec:rk-count} then determine all the pairs
inside the family that have the same length, rank, and kernel dimension, count the
resulting classes and tabulate the invariants.  Subsections \ref{subsec:z2z4-rk} and \ref{subsec:z8-rk} do the same against the two comparison families. 

\subsection{The comparison families}\label{subsec:comparison-families}

We recall the two constructions and the invariants that we will use. Throughout, a code with two superscripts is $\Z_2\Z_4$-linear, one with three
superscripts and no bar is $\Z_2\Z_4\Z_8$-linear, and one with a bar, as in
$\bar H^{a,b,c}$, is $\Z_8$-linear.

We begin with the $\Z_2\Z_4$-linear Hadamard codes. 
Deleting the coordinates in the $\Z_8$ part specialises the construction of
Subsection \ref{subsec:construction} to the recursive construction of the
$\Z_2\Z_4$-additive Hadamard codes of type $(\alpha_1,\alpha_2;t_2,t_3)$ with
$\alpha_1\not=0$ and $\alpha_2\not=0$ given in \cite{PRV06,RSV08}.  Indeed,
both \eqref{eq:construction-order8} and \eqref{eq:construction-order4} become
\begin{equation}\label{eq:z2z4-construction-order4}
A^{\ell,1}=\left(\begin{array}{cc|ccccc}
A_1&A_1&M_1&A_2&A_2&A_2&A_2\\
\zero&\one&\one&\zero&\one&\two&\mathbf{3}
\end{array}\right),
\end{equation}
where $A^{\ell-1,1}=(A_1\mid A_2)$ and $M_1=2A_1$ up to a permutation of
columns, and turns \eqref{eq:construction-order2} into
\begin{equation}\label{eq:z2z4-construction-order2}
A^{t_2,\ell}=\left(\begin{array}{cc|cc}
A_1&A_1&A_2&A_2\\
\zero&\one&\zero&\two
\end{array}\right),
\end{equation}
where $A^{t_2,\ell-1}=(A_1\mid A_2)$.  Starting from
\begin{equation}\label{eq:z2z4-base}
A^{1,1}=\left(\begin{array}{cc|c}
1&1&2\\
0&1&1
\end{array}\right)
\end{equation}
and applying \eqref{eq:z2z4-construction-order4} and
\eqref{eq:z2z4-construction-order2}, we obtain the matrices
$A^{t_2,t_3}$ of \cite{PRV06,RSV08}.  We write $\cH^{t_2,t_3}$ for the code
they generate and $H^{t_2,t_3}=\Phi(\cH^{t_2,t_3})$ for its Gray image.

The same induction as in Proposition \ref{prop:column-description} describes the columns of $A^{t_2,1}$: after deleting the first entry, the $\Z_2$ part consists of all the vectors of
$\Z_2^{t_2}$ and the $\Z_4$ part of all the normalized primitive vectors of
$\Z_4^{t_2}$.  Hence, $\alpha_1=2^{t_2+t_3-1}$ and
$\alpha_1+2\alpha_2=4^{t_2}2^{t_3-1}$, so the binary length is $2^t$ with
\begin{equation}\label{eq:z2z4-length}
t+1=2t_2+t_3.
\end{equation}
By \cite{PRV06}, $H^{t_2,t_3}$ is linear if and only if $t_2=1$, and, for
the nonlinear ones,
\begin{equation}\label{eq:z2z4-invariants}
\kernel\bigl(H^{t_2,t_3}\bigr)=t_2+t_3,
\qquad
\rank\bigl(H^{t_2,t_3}\bigr)=t_3+2t_2+\binom{t_2}{2}.
\end{equation}
Finally, by \cite{KV2015}, every $\Z_4$-linear Hadamard code is equivalent to
a $\Z_2\Z_4$-linear Hadamard code with $\alpha_1\not=0$ and $\alpha_2\not=0$.
In all the comparisons below, it therefore suffices to consider the family
$H^{t_2,t_3}$ with $t_2\geq2$ and $t_3\geq1$, the $\Z_4$-linear Hadamard codes
being automatically covered.  To avoid ambiguity, we write the parameters of a
$\Z_2\Z_4$-linear Hadamard code as $H^{U,V}$ whenever it is compared with a
code $H^{t_1,t_2,t_3}$.

We turn to the $\Z_8$-linear Hadamard codes We recall the construction of \cite{KernelZ2s}, specialised to $s=3$, which is the only case in which the rank of the comparison code is known.  Put
$T_1=\Z_8$, $T_2=2\Z_8$ and $T_3=4\Z_8$. Thus, $T_j$ is the ideal of $\Z_8$
whose elements have order at most $2^{4-j}$.  For integers $a\geq1$ and $b,c\geq0$, let $\bar A^{a,b,c}$ be the matrix whose columns are exactly all the vectors $\zz^T$ with
\begin{equation}\label{eq:z8-columns}
\zz\in\{1\}\times T_1^{a-1}\times T_2^{b}\times T_3^{c}.
\end{equation}
A row whose entries are taken from $T_j$ is said to have
\emph{level} $j$ and has order $2^{4-j}$. Thus, the first $a$ rows have level
$1$, the next $b$ have level $2$ and the last $c$ have level $3$, and the
first row is $\one$.  Let $\bar\cH^{a,b,c}$ be the $\Z_8$-additive code
generated by $\bar A^{a,b,c}$ and $\bar H^{a,b,c}=\Phi_3(\bar\cH^{a,b,c})$ its
Gray image.  By \cite{KernelZ2s}, $\bar H^{a,b,c}$ is a binary Hadamard code of
length $2^t$ with
\begin{equation}\label{eq:z8-length}
t+1=3a+2b+c.
\end{equation}
It is linear if and only if $(a,b)\in\{(1,0),(1,1)\}$. For a nonlinear
one, that is, for $a\geq2$ or $b\geq2$, its kernel has dimension
\begin{equation}\label{eq:z8-kernel}
\kernel\bigl(\bar H^{a,b,c}\bigr)=\sigma+a+b+c,
\qquad\text{where }\
\sigma=\begin{cases}
1&\text{if }a\geq2,\\
2&\text{if }a=1,
\end{cases}
\end{equation}
and its rank was determined in \cite{fernandez2019mathbb}:
\begin{equation}\label{eq:z8-rank}
\rank\bigl(\bar H^{a,b,c}\bigr)
=\frac{a^4-2a^3+35a^2+14a}{24}+\frac{b}{2}\bigl(a^2+a+b+1\bigr)+c+1.
\end{equation}

\begin{remark}\label{rem:only-s-three}
The same construction produces the $\Z_{2^s}$-linear Hadamard codes
$\bar H^{a_1,\dots,a_s}$ for every $s\geq2$, whose length and kernel are known
in general \cite{KernelZ2s}. No formula for their rank is available for
$s\geq4$, however, so a comparison by means of the rank is simply not possible for
those alphabets and we do not attempt one; it would require an invariant
that does not depend on a rank formula. Here, we compare $H^{t_1,t_2,t_3}$ with the
$\Z_4$-linear, the $\Z_2\Z_4$-linear and the $\Z_8$-linear Hadamard codes,
which is exactly the range in which both classical invariants are computable
on both sides.
\end{remark}

\subsection{The reduced rank}\label{subsec:reduced-rank}

All the comparisons below have the same shape: two Hadamard codes of the same
length $2^t$ and the same kernel dimension $\kappa$ are given, and we must
decide whether their ranks agree.  It is therefore natural to subtract from
the rank whatever depends on $t$ and $\kappa$ only and to compare what is
left.  We begin with the polynomial that carries the whole information.

\begin{definition}\label{def:Theta}
For integers $n\geq0$ and $m\geq0$, put
\begin{equation}\label{eq:Theta}
\Theta(n,m)=4n+4\binom{n}{2}+2\binom{n}{3}+\binom{n}{4}
+m\binom{n+2}{2}+\binom{m+1}{2}.
\end{equation}
\end{definition}

Comparing \eqref{eq:Theta} with \eqref{eq:rank-binomial}, we see at once that
\begin{equation}\label{eq:rank-Theta-mixed}
\rank\bigl(H^{t_1,t_2,t_3}\bigr)=\Theta(t_1,t_2)+t_3-1.
\end{equation}
The next lemma says that the same polynomial computes the ranks of the two
comparison families. This is more than a curiosity: it is what makes the
three comparisons below instances of one single computation.

\begin{lemma}\label{lem:Theta-all-families}
Let $H^{U,V}$ be a nonlinear $\Z_2\Z_4$-linear Hadamard code, that is, with
$U\geq2$ and $V\geq1$. Let $\bar H^{a,b,c}$ be a $\Z_8$-linear Hadamard
code, with $a\geq1$ and $b,c\geq0$.  Then,
\begin{equation}\label{eq:rank-Theta-z2z4}
\rank\bigl(H^{U,V}\bigr)=\Theta(0,U)+V
\end{equation}
and 
\begin{equation}\label{eq:rank-Theta-z8}
\rank\bigl(\bar H^{a,b,c}\bigr)=\Theta(a-1,b)+a+c+2.
\end{equation}
\end{lemma}

\begin{proof}
For \eqref{eq:rank-Theta-z2z4}, $\Theta(0,U)=U\binom22+\binom{U+1}{2}=2U+\binom U2$. This is exactly the second identity of \eqref{eq:z2z4-invariants}. For \eqref{eq:rank-Theta-z8}, put $n=a-1$ in the
expansion $4n+4\binom n2+2\binom n3+\binom n4=(n^4+2n^3+35n^2+58n)/24$ used in
the proof of Theorem \ref{thm:rank-main}; a direct substitution gives
$(a^4-2a^3+35a^2-10a-24)/24$.  Moreover,
$b\binom{a+1}{2}+\binom{b+1}{2}=\frac b2(a^2+a+b+1)$, and hence
$$\Theta(a-1,b)+a+c+2=\frac{a^4-2a^3+35a^2-10a-24}{24}+\frac b2\bigl(a^2+a+b+1\bigr)+a+c+2.$$
Since $\frac{-10a-24}{24}+a+2=\frac{14a+24}{24}=\frac{14a}{24}+1$, the
right-hand side is \eqref{eq:z8-rank}.
\end{proof}

We can now define the reduced rank.  Recall that a Hadamard code of length
$2^t$ has $2^{t+1}$ codewords, and that all the codes considered here contain
$\zero$, so that their kernels are linear subcodes.

\begin{definition}\label{def:reduced-rank}
Let $C$ be a binary Hadamard code of length $2^t$ with $\zero\in C$, and put
$\kappa=\kernel(C)$ and $u=t+1-\kappa$.  The \emph{reduced rank} of $C$ is
\begin{equation}\label{eq:reduced-rank}
\nu(C)=\rank(C)-\bigl(t+\Theta(0,u)-2u\bigr).
\end{equation}
\end{definition}

Because $t$ and $\kappa$ determine $u$, the subtracted quantity
$t+\Theta(0,u)-2u$ depends only on the length and on the dimension of the
kernel. The following statement, which is the only property of $\nu$ that we
will use, therefore holds by construction.

\begin{lemma}\label{lem:reduced-rank-principle}
Let $C$ and $D$ be binary Hadamard codes of the same length $2^t$ with
$\zero\in C$, $\zero\in D$ and $\kernel(C)=\kernel(D)$.  Then,
$\rank(C)=\rank(D)$ if and only if $\nu(C)=\nu(D)$.
\end{lemma}

The point of the definition is that the reduced ranks of our three families
are given by very small expressions.  They are all built from the polynomial
introduced next.

\begin{definition}\label{def:Upsilon}
For an integer $u\geq0$ and an integer $k\geq0$, put
\begin{equation}\label{eq:Upsilon}
\Upsilon_u(k)=2\binom{k}{3}+\binom{k}{4}+(u-2k)\binom{k}{2}.
\end{equation}
\end{definition}
Note that $\Upsilon_u(0)=\Upsilon_u(1)=0$ for every $u$, because all the
binomial coefficients occurring in \eqref{eq:Upsilon} vanish for $k\leq1$.
Note also that, if $u=2n+m$, then
\begin{equation}\label{eq:Upsilon-Xi}
\Upsilon_u(n)=2\binom{n}{3}+\binom{n}{4}+m\binom{n}{2},
\end{equation}
which is the form in which $\Upsilon$ will usually be evaluated.

\begin{proposition}\label{prop:reduced-ranks}
The reduced ranks of the three families are the following.
\begin{enumerate}[label=\textup{(\roman*)}]
\item If $H^{t_1,t_2,t_3}$ is nonlinear and $u=2t_1+t_2$, then
\begin{equation}\label{eq:nu-mixed}
\nu\bigl(H^{t_1,t_2,t_3}\bigr)=\Upsilon_u(t_1)
=2\binom{t_1}{3}+\binom{t_1}{4}+t_2\binom{t_1}{2}.
\end{equation}
\item If $H^{U,V}$ is a nonlinear $\Z_2\Z_4$-linear Hadamard code, or if
$\bar H^{1,b,c}$ is a nonlinear $\Z_8$-linear Hadamard code with $a=1$, that
is, with $b\geq2$, then
\begin{equation}\label{eq:nu-z2z4}
\nu\bigl(H^{U,V}\bigr)=\nu\bigl(\bar H^{1,b,c}\bigr)=1 .
\end{equation}
\item If $\bar H^{a,b,c}$ is a nonlinear $\Z_8$-linear Hadamard code with
$a\geq2$, and $u=2a+b-1$, then
\begin{equation}\label{eq:nu-z8}
\nu\bigl(\bar H^{a,b,c}\bigr)
=\Upsilon_u(a-1)-\binom{a-1}{2}+a+1-u.
\end{equation}
\end{enumerate}
\end{proposition}

\begin{proof}
Throughout, we use the identity
\begin{equation}\label{eq:Theta-shift}
\Theta(0,w)-\Theta(0,w-1)=w+1,
\end{equation}
with $w\geq1$. We also use the following reformulation of $\Theta$, valid for all
$n,m\geq0$: if $w=2n+m$ then
\begin{equation}\label{eq:Theta-split}
\Theta(n,m)=\Theta(0,w)-n+\Upsilon_w(n).
\end{equation}
To prove it, note first that
$\binom{w+1}{2}-\binom{m+1}{2}=\frac{(2n+m+1)(2n+m)-(m+1)m}{2}=2n^2+2nm+n$.
Hence, using \eqref{eq:Theta}, together with
$\Theta(0,w)=w+\binom{w+1}{2}$ and
$4\binom{n}{2}=2n^2-2n$, we obtain
\begin{align*}
\Theta(0,w)-\Theta(n,m)
&=(2n+m)+\binom{w+1}{2}-4n-4\binom n2-2\binom n3-\binom n4 \\
&\qquad{}-m\binom{n+2}{2}-\binom{m+1}{2}\\
&=n+m+2nm-2\binom n3-\binom n4-\tfrac m2(n^2+3n+2).
\end{align*}
The coefficient of $m$ in the last line is
$1+2n-\tfrac12(n^2+3n+2)=\tfrac12(n-n^2)=-\binom n2$, hence
$$\Theta(0,w)-\Theta(n,m)=n-2\binom n3-\binom n4-m\binom n2=n-\Upsilon_w(n)$$ by
\eqref{eq:Upsilon-Xi}, which is \eqref{eq:Theta-split}.

\emph{Item} (i).  Here, $\kappa=t_1+t_2+t_3$ and $t+1=3t_1+2t_2+t_3$, so
$u=t+1-\kappa=2t_1+t_2$ as in \eqref{eq:u-definition}, and
$t_3=t+1-3t_1-2t_2$.  By \eqref{eq:rank-Theta-mixed} and
\eqref{eq:Theta-split} with $n=t_1$, $m=t_2$, $w=u$, we have
$$
\rank\bigl(H^{t_1,t_2,t_3}\bigr)=\Theta(t_1,t_2)+t_3-1
=\Theta(0,u)-t_1+\Upsilon_u(t_1)+t-3t_1-2t_2 .
$$
Here, $-t_1-3t_1-2t_2=-4t_1-2(u-2t_1)=-2u$, since $t_2=u-2t_1$.  Hence,
$\rank(H^{t_1,t_2,t_3})=t+\Theta(0,u)-2u+\Upsilon_u(t_1)$. Comparing with \eqref{eq:reduced-rank} gives \eqref{eq:nu-mixed}.

\emph{Item} (ii), first case.  Here, $\kappa=U+V$ and $t+1=2U+V$ by
\eqref{eq:z2z4-length}, so $u=t+1-\kappa=U$ and $V=t+1-2U=t+1-2u$.  By
\eqref{eq:rank-Theta-z2z4}, $\rank\bigl(H^{U,V}\bigr)=\Theta(0,u)+t+1-2u
=\bigl(t+\Theta(0,u)-2u\bigr)+1$, which is \eqref{eq:nu-z2z4}.

\emph{Item} (ii), second case.  Here, $a=1$ and $b\geq2$, so $\sigma=2$ by
\eqref{eq:z8-kernel} and $\kappa=2+1+b+c=3+b+c$, while $t+1=3+2b+c$.  Hence,
$u=t+1-\kappa=b$ and $c=t-2-2b=t-2-2u$. Then, \eqref{eq:rank-Theta-z8} with
$a=1$, together with $\Theta(0,b)=\Theta(0,u)$, gives
$\rank(\bar H^{1,b,c})=\Theta(0,u)+1+c+2=\Theta(0,u)+t+1-2u
=(t+\Theta(0,u)-2u)+1,$ which is again \eqref{eq:nu-z2z4}.

\emph{Item} (iii).  Here, $\sigma=1$ because $a\geq2$, so
$\kappa=1+a+b+c$ by \eqref{eq:z8-kernel}, while $t+1=3a+2b+c$ by
\eqref{eq:z8-length}.  Hence, $u=t+1-\kappa=2a+b-1$, and $c=t+1-3a-2b$.  Put
$p=a-1$; then $b=u+1-2a=u-1-2p$, and therefore $2p+b=u-1$.  By
\eqref{eq:rank-Theta-z8} and by \eqref{eq:Theta-split} applied with $n=p$,
$m=b$ and $w=2p+b=u-1$, $$\rank\bigl(\bar H^{a,b,c}\bigr)=\Theta(p,b)+a+c+2
=\Theta(0,u-1)-p+\Upsilon_{u-1}(p)+a+c+2.$$  Now,
$\Upsilon_{u-1}(p)=\Upsilon_u(p)-\binom p2$ directly from
\eqref{eq:Upsilon}, and $\Theta(0,u-1)=\Theta(0,u)-u$ by
\eqref{eq:Theta-shift}.  Moreover, $c=t+1-3a-2b$ and $b=u-1-2p=u+1-2a$ give
$a+c+2=a+t+1-3a-2b+2=t+3-2a-2(u+1-2a)=t+1+2a-2u$.  Collecting the four
contributions, $\rank\bigl(\bar H^{a,b,c}\bigr)
=\Theta(0,u)-u-p+\Upsilon_u(p)-\binom p2+t+1+2a-2u$, and, since
$-p+2a=-(a-1)+2a=a+1$, the right-hand side equals
$\bigl(t+\Theta(0,u)-2u\bigr)+\Upsilon_u(p)-\binom p2+a+1-u$, which is
\eqref{eq:nu-z8}.
\end{proof}

Item (ii) says that, as far as length, kernel dimension and rank are
concerned, a nonlinear $\Z_8$-linear Hadamard code with $a=1$ behaves exactly
like a $\Z_2\Z_4$-linear one. Indeed, comparing the parameters computed in its
two cases shows that $\bar H^{1,b,c}$ and $H^{b,\,c+3}$ agree in all three.
Whenever the case $a=1$ occurs below we therefore refer to the $\Z_2\Z_4$
computation instead of repeating it.

We close the subsection by recording the first difference of $\Upsilon_u$,
which is used both in Subsection \ref{subsec:internal-rk} and in Subsection
\ref{subsec:z8-rk}.

\begin{lemma}\label{lem:Upsilon-difference}
For every $u\geq0$ and every $k\geq0$,
\begin{equation}\label{eq:Upsilon-difference}
\Upsilon_u(k+1)-\Upsilon_u(k)=\frac{k}{6}\bigl(k^2-15k+6u-10\bigr).
\end{equation}
\end{lemma}

\begin{proof}
Using $\binom{k+1}{j}-\binom{k}{j}=\binom{k}{j-1}$ three times, together with
$\binom{k+1}{2}-\binom k2=k$,
\begin{align*}
\Upsilon_u(k+1)-\Upsilon_u(k)
&=2\binom k2+\binom k3+(u-2k-2)\binom{k+1}{2}-(u-2k)\binom{k}{2}\\
&=2\binom k2+\binom k3+(u-2k)k-k(k+1).
\end{align*}
Since $2\binom k2=k(k-1)$ and $\binom k3=k(k-1)(k-2)/6$, the right-hand side
equals $\frac k6\bigl[6(k-1)+(k-1)(k-2)+6u-18k-6\bigr]=\frac k6\bigl[6k-6+k^2-3k+2+6u-18k-6\bigr]=\frac k6\big[k^2-15k+6u-10\bigr]$, which gives
\eqref{eq:Upsilon-difference}.
\end{proof}

\subsection{Collisions inside the family}\label{subsec:internal-rk}

We now determine every pair of distinct types of the same length whose codes
have the same rank and the same kernel dimension. We call such a pair
a \emph{collision}, since it is a case in which the two classical invariants
cannot decide.  The following example is the smallest one.

\begin{example}\label{ex:no-kernel-no-rank}
By Theorems \ref{thm:recalled-hadamard} and
\ref{thm:recalled-linearity-kernel}, the nonlinear
$\Z_2\Z_4\Z_8$-linear Hadamard codes of length $2^9$ are
$H^{1,1,5}$, $H^{1,2,3}$, $H^{1,3,1}$, $H^{2,0,4}$,
$H^{2,1,2}$, and $H^{3,0,1}$.  Their kernel dimensions are
$7,6,5,6,5,4$, respectively, and, by Theorem \ref{thm:rank-main}. Note that neither the kernel dimension nor the rank alone classifies these six
codes.  Indeed, the pairs $\{H^{1,2,3},H^{2,0,4}\}$ and
$\{H^{1,3,1},H^{2,1,2}\}$ have the same kernel dimension, while
$H^{1,2,3}$ and $H^{2,0,4}$ also have the same rank.  The pair
$(r,k)=(\rank,\kernel)$ does not classify them either, since
$H^{1,2,3}$ and $H^{2,0,4}$ both have $(r,k)=(15,6)$. In \cite{Z2Z4Z8Linearity}, these two codes had to be separated by a computer equivalence test.
\end{example}

\begin{theorem}\label{thm:all-collisions}
Let $H^{t_1,t_2,t_3}$ and $H^{t'_1,t'_2,t'_3}$ be two nonlinear members of
the family with distinct types, having the same length $2^t$, the same rank
and the same dimension of the kernel.  Then, after interchanging them if
necessary,
$$
(t_1,t_2,t_3)=(1,2,t-6)\ \text{ and }\ (t'_1,t'_2,t'_3)=(2,0,t-5),
\qquad t\geq7, 
$$
or
$$
(t_1,t_2,t_3)=(2,2,t-9)\ \text{ and }\ (t'_1,t'_2,t'_3)=(3,0,t-8),
\qquad t\geq10.
$$
Conversely, each of the two displayed pairs does have equal length, equal
rank and equal dimension of the kernel.
\end{theorem}

\begin{proof}
Let $\kappa$ be the common dimension of the kernel and put $u=t+1-\kappa$ as
in \eqref{eq:u-definition}.  Then, $u=2t_1+t_2=2t'_1+t'_2$, so
$t_2=u-2t_1\geq0$ and $t'_2=u-2t'_1\geq0$; in particular
\begin{equation}\label{eq:u-bound}
u\geq 2t_1\qquad\text{and}\qquad u\geq2t'_1.
\end{equation}
Moreover, $t_1\not=t'_1$: if $t_1=t'_1$, then $t_2=t'_2$ because both equal
$u-2t_1$, and then $t_3=t'_3=\kappa-t_1-t_2$, so the two types would
coincide.  Write $x=\min\{t_1,t'_1\}$ and $y=\max\{t_1,t'_1\}$, so
$1\leq x<y$. By Lemma \ref{lem:reduced-rank-principle} and item (i) of Proposition
\ref{prop:reduced-ranks}, the equality of the ranks means
\begin{equation}\label{eq:collision-Upsilon}
\Upsilon_u(x)=\Upsilon_u(y).
\end{equation}
Summing \eqref{eq:Upsilon-difference} for $k=x,x+1,\dots,y-1$ and using
\eqref{eq:collision-Upsilon}, we obtain
$0=\Upsilon_u(y)-\Upsilon_u(x)
=\sum_{j=x}^{y-1}\frac{j}{6}\bigl(j^2-15j+6u-10\bigr).$
Multiplying by $6$ and separating the term containing $u$ gives
\begin{equation}\label{eq:u-average}
6u\sum_{j=x}^{y-1}j
=\sum_{j=x}^{y-1}j\bigl(15j+10-j^2\bigr).
\end{equation}

We first bound $y$.  For every integer $j$, we have
$66-(15j+10-j^2)=j^2-15j+56=(j-7)(j-8)\geq0$, so $15j+10-j^2\leq66$ and the
right-hand side of \eqref{eq:u-average} is at most $66\sum_{j=x}^{y-1}j$.  On
the other hand $u\geq2y$ by \eqref{eq:u-bound}, so the left-hand side is at
least $12y\sum_{j=x}^{y-1}j$, and $\sum_{j=x}^{y-1}j$ is a positive integer
because $x<y$.  Therefore, $12y\leq66$, so $y\leq5$, and only the ten pairs
$(x,y)$ with $1\leq x<y\leq5$ remain.

For each of these ten pairs, \eqref{eq:u-average} determines $u$ uniquely,
because $\sum_{j=x}^{y-1}j>0$.  Writing $f(j)=j(15j+10-j^2)$, we have
$f(1)=24$, $f(2)=72$, $f(3)=138$ and
$f(4)=216$. The corresponding ten values of $u$ are
$$
\begin{array}{c|cccccccccc}
(x,y) & (1,2) & (1,3) & (1,4) & (1,5) & (2,3) & (2,4) & (2,5) & (3,4)
 & (3,5) & (4,5)\\[1mm]
 \hline
u & 4 & \tfrac{16}{3} & \tfrac{13}{2} & \tfrac{15}{2} & 6 & 7 & \tfrac{71}{9}
 & \tfrac{23}{3} & \tfrac{59}{7} & 9
\end{array}
$$
Only four of these values are integers, namely $u=4$, $6$, $7$, and $9$,
corresponding to $(x,y)=(1,2)$, $(2,3)$, $(2,4)$, and $(4,5)$,
respectively.  The last two are excluded by the condition $u\geq2y$ in
\eqref{eq:u-bound}, since $7<8$ and $9<10$.  Hence, exactly two possibilities
remain: $(x,y,u)=(1,2,4)$ and $(x,y,u)=(2,3,6)$.
In the first case, $t_1=1$, $t_2=u-2=2$, $t'_1=2$, and
$t'_2=u-4=0$.  Moreover, $\kappa=t+1-u=t-3$, so
$t_3=\kappa-t_1-t_2=t-6$ and
$t'_3=\kappa-t'_1-t'_2=t-5$.  The condition $t_3\geq1$ then gives
$t\geq7$. In the second case $t_1=2$, $t_2=u-4=2$, $t'_1=3$, $t'_2=u-6=0$ and $\kappa=t-5$, so
$t_3=t-9$ and $t'_3=t-8$.  The condition $t_3\geq1$ gives $t\geq10$.  This proves the direct statement.

For the converse, Corollary \ref{cor:rank-special-values} gives the two equalities 
$\rank(H^{1,2,t-6})=\rank(H^{2,0,t-5})=t+6$ and
$\rank(H^{2,2,t-9})=\rank(H^{3,0,t-8})=t+17$, while
$\kernel(H^{1,2,t-6})=\kernel(H^{2,0,t-5})=t-3$ and
$\kernel(H^{2,2,t-9})=\kernel(H^{3,0,t-8})=t-5$.  Moreover, all four codes
have length $2^t$ by \eqref{eq:length-relation}.
\end{proof}

\begin{example}
The two collision families of Theorem \ref{thm:all-collisions} start at
$t=7$ and $t=10$ respectively. Their first members are $(H^{1,2,1},H^{2,0,2})$, with $(r,k)=(13,4)$ and length $2^7$, and $(H^{2,2,1},H^{3,0,2})$, with $(r,k)=(27,5)$ and length
$2^{10}$.  Both pairs appear in \cite[Tables 2--4]{Z2Z4Z8Linearity} and both had to be
separated there with {\sc Magma}.  Theorem \ref{thm:all-collisions} shows that
these two families, continued for every larger $t$, are the only
obstructions inside the family: for any other pair of distinct types of the
same length, the rank or the dimension of the kernel already differs.
\end{example}

\subsection{The number of classes separated by the two invariants}
\label{subsec:rk-count}

We first count the admissible types, and then subtract the collisions.

\begin{proposition}\label{prop:count-types}
For $t\geq3$, let $\cA_t$ be the number of triples $(t_1,t_2,t_3)$ with
$t_1\geq1$, $t_2\geq0$, $t_3\geq1$ and $3t_1+2t_2+t_3=t+1$, that is, the
number of codes $H^{t_1,t_2,t_3}$ of length $2^t$.  Then,
\begin{equation}\label{eq:number-types}
\cA_t=\sum_{j=1}^{\lfloor t/3\rfloor}
\left(\left\lfloor\frac{t-3j}{2}\right\rfloor+1\right)
=\left\lfloor\frac{t^2+6}{12}\right\rfloor.
\end{equation}
\end{proposition}

\begin{proof}
Fix $t_1=j\geq1$.  By \eqref{eq:length-relation}, we have
$t_3=t+1-3j-2t_2$, and the condition $t_3\geq1$ is equivalent to
$t_2\leq(t-3j)/2$.  Since $t_2$ is a nonnegative integer, it takes exactly
the values $0,1,\dots,\lfloor(t-3j)/2\rfloor$, and this set is nonempty
precisely when $3j\leq t$.  This proves the first equality in
\eqref{eq:number-types}.

For the closed form, write $t=6q+r$ with $0\leq r\leq5$ and split the sum
according to the parity of $j$.  If $j$ is even, then
$t-3j\equiv t\pmod 2$, and hence
$\lfloor(t-3j)/2\rfloor=(t-3j-\varepsilon)/2$, where
$\varepsilon\in\{0,1\}$ is the parity of $t$.  If $j$ is odd, then
$t-3j\equiv t+1\pmod 2$, and hence
$\lfloor(t-3j)/2\rfloor=(t-3j-\varepsilon')/2$, where
$\varepsilon'\in\{0,1\}$ is the parity of $t+1$.  Summing the resulting two
arithmetic progressions gives $$
\begin{array}{lll}
\cA_{6q}=3q^2, & \cA_{6q+1}=3q^2+q, & \cA_{6q+2}=3q^2+2q,\\[1mm]
\cA_{6q+3}=3q^2+3q+1,\quad & \cA_{6q+4}=3q^2+4q+1,\quad &
\cA_{6q+5}=3q^2+5q+2.
\end{array}
$$
A direct check shows that each of the six expressions equals
$\lfloor((6q+r)^2+6)/12\rfloor$.  For instance, for $r=0$, we obtain
$\cA_{6q}=3q^2$ and
$\lfloor(36q^2+6)/12\rfloor=3q^2$, while for $r=3$, we obtain
$\cA_{6q+3}=3q^2+3q+1$ and
$\lfloor(36q^2+36q+15)/12\rfloor=3q^2+3q+1$.
\end{proof}

\begin{corollary}\label{cor:rk-classes}
For $t\geq3$, the number of distinct values taken by the pair
$(r,k)=(\rank,\kernel)$ on the $\cA_t$ codes $H^{t_1,t_2,t_3}$ of length $2^t$
is
\begin{equation}\label{eq:rk-classes}
\left\lfloor\frac{t^2+6}{12}\right\rfloor\ \text{ for } 3\leq t\leq6,
\qquad
\left\lfloor\frac{t^2+6}{12}\right\rfloor-1\ \text{ for } 7\leq t\leq9,
\qquad
\left\lfloor\frac{t^2+6}{12}\right\rfloor-2\ \text{ for } t\geq10.
\end{equation}
In particular, the pair $(r,k)$ classifies the whole family of length $2^t$ if
and only if $3\leq t\leq6$, in which case the number of nonequivalent codes
$H^{t_1,t_2,t_3}$ of length $2^t$ is $\cA_t$, that is, $1,1,2$ and $3$ for
$t=3,4,5,6$ respectively.
\end{corollary}

\begin{proof}
The linear member $H^{1,0,t-2}$ has $r=k=t+1$ and every nonlinear member has
$r>k$, so the linear member contributes a value of $(r,k)$ of its own.  Among
the nonlinear members, two distinct types give the same pair $(r,k)$ exactly
when they form one of the two collisions of Theorem
\ref{thm:all-collisions}, and each collision consists of exactly two types.
The two collisions never share a type, because the first involves
$t_1\in\{1,2\}$ with $t_2\in\{2,0\}$ and $\kappa=t-3$, while the second
involves $\kappa=t-5$.  Each of them therefore merges exactly two of the
$\cA_t$ types into a single value of $(r,k)$ and reduces the count by one.
Since the first exists precisely when $t\geq7$ and the second precisely when
$t\geq10$, the count is $\cA_t$ for $3\leq t\leq6$, then $\cA_t-1$ for $7\leq
t\leq9$, and then $\cA_t-2$ for $t\geq10$, which is \eqref{eq:rk-classes}.

Finally, for $t\leq6$ distinct types have distinct pairs $(r,k)$.  Since
$(r,k)$ is an invariant of the equivalence class, the codes are then pairwise
nonequivalent, and there are $\cA_t$ of them.  The values $\cA_3=\cA_4=1$,
$\cA_5=2$ and $\cA_6=3$ follow from \eqref{eq:number-types}.
\end{proof}

For $3\leq t\leq15$, formula \eqref{eq:number-types} gives
$\cA_t=1,1,2,3,4,5,7,8,10,12,14,16,19$, in agreement with the number of types
listed in \cite{Z2Z4Z8Linearity}, while Corollary \ref{cor:rk-classes} gives
$1,1,2,3,3,4,6,6,8,10,12,14,17$ distinct values of $(r,k)$, respectively.  The
two counts start to differ at $t=7$, which is exactly the length at which the
first collision appears.  Nevertheless, for $5\leq t\leq11$, the codes of the
same length were shown in \cite{Z2Z4Z8Linearity} to be pairwise
nonequivalent; in the collision cases, nonequivalence was established there by
computer equivalence tests using {\sc Magma}.  Thus, the collisions reflect a
failure of the pair $(r,k)$ to distinguish all the equivalence classes, rather
than an equivalence between the corresponding codes.

Table \ref{tab:invariants} displays the individual codes, their invariants and
the two counts for $5\leq t\leq15$.  It extends Tables 2--4 of
\cite{Z2Z4Z8Linearity}, with one essential difference: the ranks now come from
Theorem \ref{thm:rank-main} instead of from a computer computation, so the
table can be continued to any length without further work.  The collisions are
marked, and they are exactly the ones predicted by Theorem
\ref{thm:all-collisions}.

\begin{table}[ht]
\centering
\scriptsize
\setlength{\tabcolsep}{2.5pt}
\begin{tabular}{|cc|cc|cc|cc|cc|cc|}
\hline
\multicolumn{2}{|c|}{$t=5$} & \multicolumn{2}{c|}{$t=6$} & \multicolumn{2}{c|}{$t=7$} & \multicolumn{2}{c|}{$t=8$} & \multicolumn{2}{c|}{$t=9$} & \multicolumn{2}{c|}{$t=10$}\\
\hline
Code & $(r,k)$ & Code & $(r,k)$ & Code & $(r,k)$ & Code & $(r,k)$ & Code & $(r,k)$ & Code & $(r,k)$\\[0.2em]
\hline
$H^{1,0,3}{}^{*}$ & $(6,6)$ & $H^{1,0,4}{}^{*}$ & $(7,7)$ & $H^{1,0,5}{}^{*}$ & $(8,8)$ & $H^{1,0,6}{}^{*}$ & $(9,9)$ & $H^{1,0,7}{}^{*}$ & $(10,10)$ & $H^{1,0,8}{}^{*}$ & $(11,11)$\\
$H^{1,1,1}$ & $(8,3)$ & $H^{1,1,2}$ & $(9,4)$ & $H^{1,1,3}$ & $(10,5)$ & $H^{1,1,4}$ & $(11,6)$ & $H^{1,1,5}$ & $(12,7)$ & $H^{1,1,6}$ & $(13,8)$\\
{} & {} & $H^{2,0,1}$ & $(12,3)$ & $H^{1,2,1}{}^{\bowtie}$ & $(13,4)$ & $H^{1,2,2}{}^{\bowtie}$ & $(14,5)$ & $H^{1,2,3}{}^{\bowtie}$ & $(15,6)$ & $H^{1,2,4}{}^{\bowtie}$ & $(16,7)$\\
{} & {} & {} & {} & $H^{2,0,2}{}^{\bowtie}$ & $(13,4)$ & $H^{2,0,3}{}^{\bowtie}$ & $(14,5)$ & $H^{1,3,1}$ & $(19,5)$ & $H^{1,3,2}$ & $(20,6)$\\
{} & {} & {} & {} & {} & {} & $H^{2,1,1}$ & $(19,4)$ & $H^{2,0,4}{}^{\bowtie}$ & $(15,6)$ & $H^{2,0,5}{}^{\bowtie}$ & $(16,7)$\\
{} & {} & {} & {} & {} & {} & {} & {} & $H^{2,1,2}$ & $(20,5)$ & $H^{2,1,3}$ & $(21,6)$\\
{} & {} & {} & {} & {} & {} & {} & {} & $H^{3,0,1}$ & $(26,4)$ & $H^{2,2,1}{}^{\bowtie}$ & $(27,5)$\\
{} & {} & {} & {} & {} & {} & {} & {} & {} & {} & $H^{3,0,2}{}^{\bowtie}$ & $(27,5)$\\
\hline
$\cA_t$ & $2$ & $\cA_t$ & $3$ & $\cA_t$ & $4$ & $\cA_t$ & $5$ & $\cA_t$ & $7$ & $\cA_t$ & $8$\\
$\#(r,k)$ & $2$ & $\#(r,k)$ & $3$ & $\#(r,k)$ & $3$ & $\#(r,k)$ & $4$ & $\#(r,k)$ & $6$ & $\#(r,k)$ & $6$\\
\hline
\end{tabular}
\medskip

\begin{tabular}{|cc|cc|cc|cc|cc|}
\hline
\multicolumn{2}{|c|}{$t=11$} & \multicolumn{2}{c|}{$t=12$} & \multicolumn{2}{c|}{$t=13$} & \multicolumn{2}{c|}{$t=14$} & \multicolumn{2}{c|}{$t=15$}\\
\hline
Code & $(r,k)$ & Code & $(r,k)$ & Code & $(r,k)$ & Code & $(r,k)$ & Code & $(r,k)$\\[0.2em]
\hline
$H^{1,0,9}{}^{*}$ & $(12,12)$ & $H^{1,0,10}{}^{*}$ & $(13,13)$ & $H^{1,0,11}{}^{*}$ & $(14,14)$ & $H^{1,0,12}{}^{*}$ & $(15,15)$ & $H^{1,0,13}{}^{*}$ & $(16,16)$\\
$H^{1,1,7}$ & $(14,9)$ & $H^{1,1,8}$ & $(15,10)$ & $H^{1,1,9}$ & $(16,11)$ & $H^{1,1,10}$ & $(17,12)$ & $H^{1,1,11}$ & $(18,13)$\\
$H^{1,2,5}{}^{\bowtie}$ & $(17,8)$ & $H^{1,2,6}{}^{\bowtie}$ & $(18,9)$ & $H^{1,2,7}{}^{\bowtie}$ & $(19,10)$ & $H^{1,2,8}{}^{\bowtie}$ & $(20,11)$ & $H^{1,2,9}{}^{\bowtie}$ & $(21,12)$\\
$H^{1,3,3}$ & $(21,7)$ & $H^{1,3,4}$ & $(22,8)$ & $H^{1,3,5}$ & $(23,9)$ & $H^{1,3,6}$ & $(24,10)$ & $H^{1,3,7}$ & $(25,11)$\\
$H^{1,4,1}$ & $(26,6)$ & $H^{1,4,2}$ & $(27,7)$ & $H^{1,4,3}$ & $(28,8)$ & $H^{1,4,4}$ & $(29,9)$ & $H^{1,4,5}$ & $(30,10)$\\
$H^{2,0,6}{}^{\bowtie}$ & $(17,8)$ & $H^{2,0,7}{}^{\bowtie}$ & $(18,9)$ & $H^{1,5,1}$ & $(34,7)$ & $H^{1,5,2}$ & $(35,8)$ & $H^{1,5,3}$ & $(36,9)$\\
$H^{2,1,4}$ & $(22,7)$ & $H^{2,1,5}$ & $(23,8)$ & $H^{2,0,8}{}^{\bowtie}$ & $(19,10)$ & $H^{2,0,9}{}^{\bowtie}$ & $(20,11)$ & $H^{1,6,1}$ & $(43,8)$\\
$H^{2,2,2}{}^{\bowtie}$ & $(28,6)$ & $H^{2,2,3}{}^{\bowtie}$ & $(29,7)$ & $H^{2,1,6}$ & $(24,9)$ & $H^{2,1,7}$ & $(25,10)$ & $H^{2,0,10}{}^{\bowtie}$ & $(21,12)$\\
$H^{3,0,3}{}^{\bowtie}$ & $(28,6)$ & $H^{2,3,1}$ & $(36,6)$ & $H^{2,2,4}{}^{\bowtie}$ & $(30,8)$ & $H^{2,2,5}{}^{\bowtie}$ & $(31,9)$ & $H^{2,1,8}$ & $(26,11)$\\
$H^{3,1,1}$ & $(37,5)$ & $H^{3,0,4}{}^{\bowtie}$ & $(29,7)$ & $H^{2,3,2}$ & $(37,7)$ & $H^{2,3,3}$ & $(38,8)$ & $H^{2,2,6}{}^{\bowtie}$ & $(32,10)$\\
{} & {} & $H^{3,1,2}$ & $(38,6)$ & $H^{3,0,5}{}^{\bowtie}$ & $(30,8)$ & $H^{2,4,1}$ & $(46,7)$ & $H^{2,3,4}$ & $(39,9)$\\
{} & {} & $H^{4,0,1}$ & $(49,5)$ & $H^{3,1,3}$ & $(39,7)$ & $H^{3,0,6}{}^{\bowtie}$ & $(31,9)$ & $H^{2,4,2}$ & $(47,8)$\\
{} & {} & {} & {} & $H^{3,2,1}$ & $(49,6)$ & $H^{3,1,4}$ & $(40,8)$ & $H^{3,0,7}{}^{\bowtie}$ & $(32,10)$\\
{} & {} & {} & {} & $H^{4,0,2}$ & $(50,6)$ & $H^{3,2,2}$ & $(50,7)$ & $H^{3,1,5}$ & $(41,9)$\\
{} & {} & {} & {} & {} & {} & $H^{4,0,3}$ & $(51,7)$ & $H^{3,2,3}$ & $(51,8)$\\
{} & {} & {} & {} & {} & {} & $H^{4,1,1}$ & $(65,6)$ & $H^{3,3,1}$ & $(62,7)$\\
{} & {} & {} & {} & {} & {} & {} & {} & $H^{4,0,4}$ & $(52,8)$\\
{} & {} & {} & {} & {} & {} & {} & {} & $H^{4,1,2}$ & $(66,7)$\\
{} & {} & {} & {} & {} & {} & {} & {} & $H^{5,0,1}$ & $(85,6)$\\
\hline
$\cA_t$ & $10$ & $\cA_t$ & $12$ & $\cA_t$ & $14$ & $\cA_t$ & $16$ & $\cA_t$ & $19$\\
$\#(r,k)$ & $8$ & $\#(r,k)$ & $10$ & $\#(r,k)$ & $12$ & $\#(r,k)$ & $14$ & $\#(r,k)$ & $17$\\
\hline
\end{tabular}
\caption{Rank $r$ and dimension of the kernel $k$ of the
$\Z_2\Z_4\Z_8$-linear Hadamard codes $H^{t_1,t_2,t_3}$ of length $2^t$, for
$5\leq t\leq15$.  The unique linear code of each length is marked with an
asterisk.  The codes marked with $\bowtie$ form the collision pairs of Theorem
\ref{thm:all-collisions}, that is, the pairs having the same rank and the same
kernel dimension.  The last two rows give the number $\cA_t$ of types and the
number $\#(r,k)$ of distinct values of $(r,k)$ for each $t$.}
\label{tab:invariants}
\end{table}

\subsection{Comparison with the $\Z_4$-linear and the $\Z_2\Z_4$-linear
Hadamard codes}\label{subsec:z2z4-rk}

We now compare the family with the two smaller classical families.  By
\cite{KV2015}, it suffices to treat the $\Z_2\Z_4$-linear Hadamard codes
$H^{U,V}$ with $U\geq2$ and $V\geq1$, as explained in Subsection
\ref{subsec:comparison-families}.  The first lemma determines the parameters
of the only possible partner.

\begin{lemma}\label{lem:z2z4-partner}
Let $H^{t_1,t_2,t_3}$ be nonlinear and let $H^{U,V}$ be a nonlinear
$\Z_2\Z_4$-linear Hadamard code with the same length and the same dimension of the kernel.  Then,
\begin{equation}\label{eq:z2z4-partner}
U=2t_1+t_2\qquad\text{and}\qquad V=t_3-t_1,
\end{equation}
and such a code exists if and only if $t_3\geq t_1+1$.
\end{lemma}

\begin{proof}
Equality of the lengths and of the kernel dimensions gives, by
\eqref{eq:length-relation}, \eqref{eq:z2z4-length} and
\eqref{eq:z2z4-invariants}, $3t_1+2t_2+t_3=2U+V$ and $t_1+t_2+t_3=U+V$.
Subtracting the second identity from the first one gives $U=2t_1+t_2$, and
then $V=t_1+t_2+t_3-U=t_3-t_1$.  A $\Z_2\Z_4$-linear Hadamard code with
$\alpha_1\not=0$ has $V\geq1$, which is the condition $t_3\geq t_1+1$; and it
is nonlinear because $U=2t_1+t_2\geq2$, with $U=2$ only when
$(t_1,t_2)=(1,0)$, which is excluded since $H^{t_1,t_2,t_3}$ is nonlinear.
Conversely, if $t_3\geq t_1+1$, the parameters \eqref{eq:z2z4-partner} are
admissible and define such a code.
\end{proof}

\begin{theorem}\label{thm:z2z4-coincidences}
Let $H^{t_1,t_2,t_3}$ be nonlinear and let $D$ be a $\Z_4$-linear or a
$\Z_2\Z_4$-linear Hadamard code of the same length.  If $H^{t_1,t_2,t_3}$ and
$D$ have the same rank and the same dimension of the kernel, then
$(t_1,t_2)=(2,1)$ and $D$ is permutation equivalent to $H^{5,t-9}$, where
$2^t$ is the common length; in particular $t\geq10$.  Conversely, for every
$t\geq10$, the two codes
$$
H^{2,1,t-7}\qquad\text{and}\qquad H^{5,t-9}
$$
have the same length $2^t$, the same dimension of the kernel $t-4$ and the
same rank $t+11$.
\end{theorem}

\begin{proof}
Every $\Z_4$-linear Hadamard code is equivalent to a $\Z_2\Z_4$-linear one
with $\alpha_1\not=0$ and $\alpha_2\not=0$ by \cite{KV2015}. The rank and
the dimension of the kernel are invariants of the equivalence class, so we may
assume that $D=H^{U, V}$ with
$U\geq2$ and $V\geq1$.  If $D$ is linear, it cannot have the same rank and
kernel dimension as the nonlinear code $H^{t_1,t_2,t_3}$, because a code is
linear exactly when its rank equals its kernel dimension; so $D$ is
nonlinear.  By Lemma \ref{lem:z2z4-partner}, $U$ and $V$ are given by
\eqref{eq:z2z4-partner}.

By Lemma \ref{lem:reduced-rank-principle} and items (i) and (ii) of
Proposition \ref{prop:reduced-ranks}, the two ranks agree if and only if
$\Upsilon_u(t_1)=1$ with $u=2t_1+t_2$, that is, by \eqref{eq:Upsilon-Xi}, if
and only if
\begin{equation}\label{eq:z2z4-criterion}
2\binom{t_1}{3}+\binom{t_1}{4}+t_2\binom{t_1}{2}=1.
\end{equation}
The three summands of \eqref{eq:z2z4-criterion} are nonnegative integers, so
exactly one of them equals $1$ and the other two vanish.  The first summand
is even, so it cannot be $1$; hence $2\binom{t_1}{3}=0$, which forces
$t_1\leq2$, and then $\binom{t_1}{4}=0$ automatically.  Consequently,
\eqref{eq:z2z4-criterion} reduces to $t_2\binom{t_1}{2}=1$, which forces
$\binom{t_1}{2}=1$ and $t_2=1$, that is, $t_1=2$ and $t_2=1$. So $(t_1,t_2)=(2,1)$.  By \eqref{eq:length-relation}, the length relation reads $t+1=6+2+t_3$, so $t_3=t-7$. By \eqref{eq:z2z4-partner}, we get $U=5$ and $V=t_3-t_1=t-9$. The condition $V\geq1$ gives $t\geq10$.  This
proves the direct statement.

For the converse, Corollary \ref{cor:rank-special-values} gives
$\rank(H^{2,1,t-7})=t+11$, while by \eqref{eq:z2z4-invariants}, we have
$\rank(H^{5,t-9})=(t-9)+10+\binom52=t+11$. Their kernel dimensions are
$2+1+(t-7)=t-4$ and $5+(t-9)=t-4$, respectively, and both codes have
length $2^t$ by \eqref{eq:length-relation} and \eqref{eq:z2z4-length}.
\end{proof}

\begin{example}
The smallest instance of Theorem \ref{thm:z2z4-coincidences} occurs at
$t=10$: the codes $H^{2,1,3}$ and $H^{5,1}$, both of length $2^{10}$,
have $(r,k)=(21,6)$.   For every other nonlinear type
the two classical invariants already separate $H^{t_1,t_2,t_3}$ from all the
$\Z_4$-linear and $\Z_2\Z_4$-linear Hadamard codes of the same length.  As an illustration, consider $t_1=1$.  In this case, the criterion
\eqref{eq:z2z4-criterion} reads $0=1$ and hence cannot be satisfied.  Indeed,
by \eqref{eq:nu-mixed}, the reduced rank of $H^{1,t_2,t_3}$ is
$\Upsilon_u(1)=0$, whereas the corresponding reduced rank of a
$\Z_2\Z_4$-linear Hadamard code is $1$.  Concretely,
$\rank(H^{1,t_2,t_3})$ is always exactly one less than the rank of the
$\Z_2\Z_4$-linear Hadamard code with the same length and the same kernel
dimension.  In particular, they cannot have the
same pair $(r,k)$.
\end{example}

\subsection{Comparison with the $\Z_8$-linear Hadamard codes}
\label{subsec:z8-rk}

This comparison is the longest, because the rank formulas for both families now involve quartic terms. However, the reduced rank reduces the comparison to a single Diophantine equation, which we then solve
completely.  We first fix the parameters of the possible partner.

\begin{lemma}\label{lem:z8-partner}
Let $H^{t_1,t_2,t_3}$ be nonlinear and let $\bar H^{a,b,c}$ be a nonlinear
$\Z_8$-linear Hadamard code with the same length and the same dimension of
the kernel.
\begin{enumerate}[label=\textup{(\roman*)}]
\item If $a\geq2$, then
\begin{equation}\label{eq:z8-partner-big}
b=2t_1+t_2+1-2a\qquad\text{and}\qquad c=t_3+a-t_1-2.
\end{equation}
\item If $a=1$, then $b\geq2$ and
\begin{equation}\label{eq:z8-partner-one}
b=2t_1+t_2\qquad\text{and}\qquad c=t_3-t_1-3.
\end{equation}
\end{enumerate}
\end{lemma}

\begin{proof}
For item (i), $a\geq2$ gives $\sigma=1$ by \eqref{eq:z8-kernel}, so
\eqref{eq:z8-length} and \eqref{eq:z8-kernel} read $3t_1+2t_2+t_3=3a+2b+c$ and
$t_1+t_2+t_3=1+a+b+c$. Subtracting the second identity from the first one
gives
$2t_1+t_2=2a+b-1$, that is the first half of \eqref{eq:z8-partner-big}, and
substituting it back into the second identity gives
$c=t_1+t_2+t_3-1-a-b=t_3+a-t_1-2$.

For item (ii), a nonlinear $\Z_8$-linear Hadamard code with $a=1$ has $b\geq2$, because the
linear ones are exactly those with $(a,b)\in\{(1,0),(1,1)\}$. Hence, $\sigma=2$ by \eqref{eq:z8-kernel}, and the two identities read
$3t_1+2t_2+t_3=3+2b+c$ and
$t_1+t_2+t_3=3+b+c$. Subtracting gives $2t_1+t_2=b$
and then $c=t_1+t_2+t_3-3-b=t_3-t_1-3$.
\end{proof}

The case $a=1$ needs no new work.

\begin{proposition}\label{prop:z8-a-one}
Let $H^{t_1,t_2,t_3}$ be nonlinear and let $\bar H^{1,b,c}$ be a nonlinear
$\Z_8$-linear Hadamard code with the same length and the same dimension of
the kernel.  Then, the two ranks agree if and only if $(t_1,t_2)=(2,1)$, and
in that case $b=5$ and $c=t-12$, where $2^t$ is the common length; in
particular $t\geq12$.
\end{proposition}

\begin{proof}
By item (ii) of Proposition \ref{prop:reduced-ranks}, the reduced rank of
$\bar H^{1,b,c}$ equals $1$, and so does the reduced rank of every nonlinear
$\Z_2\Z_4$-linear Hadamard code.  Hence, by Lemma
\ref{lem:reduced-rank-principle}, the criterion is again
\eqref{eq:z2z4-criterion}, and the argument of the proof of Theorem
\ref{thm:z2z4-coincidences} gives $(t_1,t_2)=(2,1)$.  Then
\eqref{eq:z8-partner-one} gives $b=2\cdot2+1=5$ and $c=t_3-5$; and, since
$t_3=t-7$ as computed there, $c=t-12$.  The condition $c\geq0$ gives
$t\geq12$.
\end{proof}

We now treat the main case $a\geq2$.  The following lemma reduces everything
to one equation between two nonnegative integer parameters.

\begin{lemma}\label{lem:z8-criterion}
Let $H^{t_1,t_2,t_3}$ be nonlinear, let $\bar H^{a,b,c}$ be a nonlinear
$\Z_8$-linear Hadamard code with $a\geq2$ having the same length and the same
dimension of the kernel, and put
\begin{equation}\label{eq:npb}
n=t_1,\qquad m=t_2,\qquad p=a-1,\qquad b=2n+m-1-2p.
\end{equation}
Then, $n\geq1$, $m\geq0$, $p\geq1$, $b\geq0$, and the two ranks agree if and
only if
\begin{equation}\label{eq:z8-criterion}
m\Bigl(\binom{n}{2}-\theta_p\Bigr)=\Gamma_p(n),
\end{equation}
where
\begin{equation}\label{eq:theta-Lambda-Gamma}
\theta_p=\binom p2-1,\quad
\Lambda_p=2\binom p3+\binom p4+1-p,\quad
\Gamma_p(n)=\Lambda_p+\theta_p(2n-1-2p)-2\binom n3-\binom n4.
\end{equation}
\end{lemma}

\begin{proof}
By item (i) of Lemma \ref{lem:z8-partner}, we have $b=2t_1+t_2+1-2a$, which is
the value in \eqref{eq:npb}, and $b\geq0$ because $b$ is a parameter of an
existing code; also $p=a-1\geq1$.  Put $u=2n+m$; then $u=2p+b+1$ by
\eqref{eq:npb}.  By Lemma \ref{lem:reduced-rank-principle} and items (i) and
(iii) of Proposition \ref{prop:reduced-ranks}, the two ranks agree if and
only if $\Upsilon_u(n)=\Upsilon_u(p)-\binom p2+a+1-u$. By \eqref{eq:Upsilon}, we have $\Upsilon_u(p)=2\binom p3+\binom p4+(u-2p)
\binom p2$ and $u-2p=b+1$, so
$\Upsilon_u(p)-\binom p2=2\binom p3+\binom p4+b\binom p2$; and
$a+1-u=p+2-(2p+b+1)=1-p-b$.  Using \eqref{eq:Upsilon-Xi}, on the left-hand
side, the criterion becomes
\begin{equation}\label{eq:z8-criterion-raw}
2\binom n3+\binom n4+m\binom n2
=2\binom p3+\binom p4+b\binom p2+1-p-b.
\end{equation}
Now $b\binom p2-b=b\theta_p$ by \eqref{eq:theta-Lambda-Gamma}, so the
right-hand side of \eqref{eq:z8-criterion-raw} equals
$\Lambda_p+b\,\theta_p$.  Substituting $b=2n+m-1-2p$, we get $2\binom n3+\binom
n4+m\binom n2=\Lambda_p+\theta_p(2n-1-2p)+m\,\theta_p$. Moving the two
terms containing $m$ to the left and the two quartic
terms to the right gives exactly \eqref{eq:z8-criterion}.
\end{proof}

Equation \eqref{eq:z8-criterion} has a very rigid structure, which we now
exploit.  Note that
\begin{equation}\label{eq:delta-sign}
\binom n2-\theta_p=\binom n2-\binom p2+1
\end{equation}
is at most $2-p$ when $n\leq p-1$, equal to $1$ when $n=p$, and at least
$p+1$ when $n\geq p+1$.  Its sign is therefore governed by the comparison of
$n$ with $p$: it is positive for $n\geq p$ and negative for $n\leq p-1$, the
only exception being $p=2$, where $n\leq p-1$ forces $n=1$ and makes it
vanish.  The next two lemmas dispose of all the cases with $n\geq p$ and of
all the cases with $p\geq8$.

\begin{lemma}\label{lem:z8-n-at-least-p}
In the situation of Lemma \ref{lem:z8-criterion}, suppose that $n\geq p$.
Then, \eqref{eq:z8-criterion} holds if and only if $(n,p)=(1,1)$ and $m=1$.
\end{lemma}

\begin{proof}
By \eqref{eq:delta-sign}, we have $\binom n2-\theta_p\geq1>0$ for every
$n\geq p$, since $\binom n2\geq\binom p2$.  As $m\geq0$, the left-hand side
of \eqref{eq:z8-criterion} is therefore nonnegative, and it is zero exactly
when $m=0$.  It thus suffices to compute the sign of $\Gamma_p(n)$ for
$n\geq p$.

First, evaluating \eqref{eq:theta-Lambda-Gamma} at $n=p$ and using
$2p-1-2p=-1$, we have 
$$
\Gamma_p(p)=\Lambda_p-\theta_p-2\binom p3-\binom p4
=\Bigl(2\binom p3+\binom p4+1-p\Bigr)-\Bigl(\binom p2-1\Bigr)
-2\binom p3-\binom p4
=2-p-\binom p2.
$$
Hence, $\Gamma_1(1)=2-1-0=1$, while $\Gamma_p(p)=2-p-\binom p2<0$ for every
$p\geq2$.

Second, for $k\geq p$ we have, again from
\eqref{eq:theta-Lambda-Gamma}, $$\Gamma_p(k+1)-\Gamma_p(k)=2\theta_p-2\binom
k2-\binom k3 \leq 2\theta_p-2\binom p2=-2<0,$$ because $\binom k2\geq\binom p2$
and $\binom k3\geq0$.  Therefore, $\Gamma_p$ is strictly decreasing on the integers $k\geq p$.

Consequently, if $p\geq2$ then $\Gamma_p(n)\leq\Gamma_p(p)<0$ for every
$n\geq p$, so \eqref{eq:z8-criterion} would force its left-hand side to be
negative, which is impossible.  If $p=1$, then $\Gamma_1(1)=1$ and
$\Gamma_1(n)\leq\Gamma_1(2)=-1<0$ for $n\geq2$, so only $n=1$ survives.  For
$(n,p)=(1,1)$ we have $\theta_1=-1$, and the left-hand side of
\eqref{eq:z8-criterion} is $m\bigl(\binom12-\theta_1\bigr)=m$ and the
right-hand side is $\Gamma_1(1)=1$.  Hence, $m=1$.
\end{proof}

\begin{lemma}\label{lem:z8-p-bound}
In the situation of Lemma \ref{lem:z8-criterion}, suppose that $n\leq p-1$.
Then $p\leq7$.
\end{lemma}

\begin{proof}
Since $n\leq p-1$, we have $p\geq2$ and, by \eqref{eq:delta-sign},
$\binom n2-\theta_p\leq\binom{p-1}{2}-\binom p2+1=2-p\leq0$.  Because
$m\geq0$, the left-hand side of \eqref{eq:z8-criterion} is therefore at most
$0$, and, more precisely,
\begin{equation}\label{eq:z8-lower}
-\Gamma_p(n)=m\Bigl(\theta_p-\binom n2\Bigr)\geq
(2p+1-2n)\Bigl(\theta_p-\binom n2\Bigr),
\end{equation}
where we used $m\geq 2p+1-2n$, which is exactly the condition $b\geq0$ in
\eqref{eq:npb}, and $\theta_p-\binom n2\geq p-2\geq0$.  Substituting
\eqref{eq:theta-Lambda-Gamma} into \eqref{eq:z8-lower} and cancelling the
common term $(2p+1-2n)\theta_p$ on both sides, the inequality
\eqref{eq:z8-lower} becomes $2\binom n3+\binom n4-\Lambda_p\geq-(2p+1-2n)\binom n2$, that is,
\begin{equation}\label{eq:z8-necessary}
\Upsilon_{2p+1}(n)\geq\Lambda_p,
\end{equation}
where we recognised
$2\binom n3+\binom n4+(2p+1-2n)\binom n2=\Upsilon_{2p+1}(n)$ from
\eqref{eq:Upsilon}.

We now show that \eqref{eq:z8-necessary} fails when $p\geq8$.  Assume
$p\geq8$.  By Lemma \ref{lem:Upsilon-difference} with $u=2p+1$,
$\Upsilon_{2p+1}(k+1)-\Upsilon_{2p+1}(k)=\tfrac k6(k^2-15k+12p-4)$, and the
quadratic $k^2-15k+12p-4$ is positive for every integer $k$, because
its minimum over the integers is attained at $k=7$ and at $k=8$, where its
value is $12p-60\geq36>0$.  Hence, $\Upsilon_{2p+1}$ is nondecreasing on the
nonnegative integers, so $\Upsilon_{2p+1}(n)\leq\Upsilon_{2p+1}(p-1)$
because $n\leq p-1$.  It therefore suffices to prove that
$\Upsilon_{2p+1}(p-1)<\Lambda_p$.  Using
$2\binom p3-2\binom{p-1}{3}=2\binom{p-1}{2}$ and
$\binom p4-\binom{p-1}{4}=\binom{p-1}{3}$, and
$\Upsilon_{2p+1}(p-1)=2\binom{p-1}{3}+\binom{p-1}{4}+3\binom{p-1}{2}$,
we get
\begin{align*}
\Lambda_p-\Upsilon_{2p+1}(p-1)
&=2\binom{p-1}{2}+\binom{p-1}{3}+1-p-3\binom{p-1}{2}\\
&=\binom{p-1}{3}-\binom{p-1}{2}-(p-1)
=(p-1)\left[\frac{(p-2)(p-3)}{6}-\frac{p-2}{2}-1\right]\\
&=\frac{(p-1)\bigl[(p-2)(p-3)-3(p-2)-6\bigr]}{6}
=\frac{(p-1)\bigl(p^2-8p+6\bigr)}{6}.
\end{align*}
For $p\geq8$, we have $p^2-8p+6=p(p-8)+6\geq6>0$, so
$\Lambda_p>\Upsilon_{2p+1}(p-1)\geq\Upsilon_{2p+1}(n)$, contradicting
\eqref{eq:z8-necessary}.  Hence $p\leq7$.
\end{proof}

\begin{proposition}\label{prop:z8-coincidences}
Let $H^{t_1,t_2,t_3}$ be nonlinear and let $\bar H^{a,b,c}$ be a nonlinear
$\Z_8$-linear Hadamard code of the same length $2^t$, with the same rank and
the same dimension of the kernel.  Then the pair $\bigl(H^{t_1,t_2,t_3},
\bar H^{a,b,c}\bigr)$ is one of the following five, and each of the five does have
equal length, equal rank and equal dimension of the kernel:
\begin{equation}\label{eq:z8-five}
\begin{array}{llll}
\bigl(H^{1,1,t-4},\ \bar H^{2,0,t-5}\bigr), & t\geq5, &
(r,k)=(t+3,\,t-2),\\[1mm]
\bigl(H^{2,1,t-7},\ \bar H^{1,5,t-12}\bigr), & t\geq12, &
(r,k)=(t+11,\,t-4),\\[1mm]
\bigl(H^{1,5,t-12},\ \bar H^{4,0,t-11}\bigr), & t\geq13, &
(r,k)=(t+21,\,t-6),\\[1mm]
\bigl(H^{2,6,t-17},\ \bar H^{4,3,t-17}\bigr), & t\geq18, &
(r,k)=(t+51,\,t-9),\\[1mm]
\bigl(H^{4,5,t-21},\ \bar H^{6,2,t-21}\bigr), & t\geq22, &
(r,k)=(t+117,\,t-12).
\end{array}
\end{equation}
\end{proposition}

\begin{proof}
The case $a=1$ is Proposition \ref{prop:z8-a-one} and produces exactly the
second line of \eqref{eq:z8-five}.  Assume therefore $a\geq2$ and adopt the
notation \eqref{eq:npb}; we must find all the solutions of
\eqref{eq:z8-criterion} with $n\geq1$, $m\geq0$, $p\geq1$ and
$b=2n+m-1-2p\geq0$.

By Lemma \ref{lem:z8-n-at-least-p}, the only solution with $n\geq p$ is
$(n,m,p)=(1,1,1)$.  By Lemma \ref{lem:z8-p-bound}, every solution with
$n\leq p-1$ has $p\leq7$; and in that case $n\leq6$.  It remains to examine
the finitely many pairs $(n,p)$ with $1\leq n\leq p-1\leq6$.  For each of
them, \eqref{eq:delta-sign} gives $\binom n2-\theta_p<0$, so
\eqref{eq:z8-criterion} determines $m$ uniquely, namely
$m=\Gamma_p(n)\big/\bigl(\binom n2-\theta_p\bigr)$, and the pair contributes a
solution exactly when this quotient is a
nonnegative integer satisfying $b=2n+m-1-2p\geq0$.  Table \ref{tab:z8-search}
lists, for each such pair, the numbers
$\Gamma_p(n)$ and $\binom n2-\theta_p$ and their quotient.  The quotient is a
nonnegative integer in
exactly three of them, namely $(n,p)=(1,3)$, $(2,3)$ and $(4,5)$, with $m=5$,
$m=6$ and $m=5$ respectively, and in all three cases
$b=2n+m-1-2p$ equals $0$, $3$ and $2$, which are nonnegative as required.

We have therefore found exactly four solutions with $a\geq2$, namely
$(n,m,p)$ equal to $(1,1,1)$, $(1,5,3)$, $(2,6,3)$ and $(4,5,5)$. Translating
them back through \eqref{eq:npb} and Lemma \ref{lem:z8-partner} gives
$(t_1,t_2,a,b)$ equal to $(1,1,2,0)$, $(1,5,4,0)$, $(2,6,4,3)$ and
$(4,5,6,2)$, respectively.  In each case, the length relation
\eqref{eq:length-relation} gives $t_3$ and then \eqref{eq:z8-partner-big}
gives $c$:
$$
\begin{array}{llll}
(1,1,2,0): & t_3=t-4, & c=t_3+2-1-2=t-5, & t\geq5;\\
(1,5,4,0): & t_3=t-12, & c=t_3+4-1-2=t-11, & t\geq13;\\
(2,6,4,3): & t_3=t-17, & c=t_3+4-2-2=t-17, & t\geq18;\\
(4,5,6,2): & t_3=t-21, & c=t_3+6-4-2=t-21, & t\geq22,
\end{array}
$$
where the lower bound on $t$ is the condition $t_3\geq1$, which in each case
also guarantees $c\geq0$.  These are the first, third, fourth and fifth
lines of \eqref{eq:z8-five}.

Finally, the displayed values of $(r,k)$ follow by substituting the parameters
into \eqref{eq:rank-binomial} and into item (ii) of Theorem
\ref{thm:recalled-linearity-kernel}; for instance
$\rank(H^{4,5,t-21})=(t-22)+16+24+8+1+75+15=t+117$ and
$\kernel(H^{4,5,t-21})=4+5+(t-21)=t-12$.  Since the two codes of each
line have the same length and the same dimension of the kernel by
construction, and the same reduced rank by the criterion just solved, they
have the same rank.
\end{proof}

\begin{table}[ht]
\centering
\caption{The search of the proof of Proposition \ref{prop:z8-coincidences}
over the pairs $1\leq n\leq p-1\leq 6$.  Each cell displays the fraction
$\Gamma_p(n)\big/\bigl(\binom n2-\theta_p\bigr)$ and its value, the only
candidate for $m$; it gives a solution of \eqref{eq:z8-criterion} only when it
is a nonnegative integer, and the three cells where this happens are printed
in bold.  In the single cell of the row $p=2$ the denominator is $0$ and the
numerator is $-1$, so there \eqref{eq:z8-criterion} has no solution at all.}
\label{tab:z8-search}
\footnotesize
\renewcommand{\arraystretch}{1.3}
\begin{tabular}{@{}ccc|cccccc@{}}
\toprule
$p$ & $\theta_p$ & $\Lambda_p$ & $n=1$ & $n=2$ & $n=3$ & $n=4$ & $n=5$ &
$n=6$\\
\midrule
$2$ & $0$ & $-1$ & $\frac{-1}{0}$ &  &  &  &  & \\
$3$ & $2$ & $0$ & $\frac{-10}{-2}{=}\mathbf{5}$ & $\frac{-6}{-1}{=}\mathbf{6}$ &  &  &  & \\
$4$ & $5$ & $6$ & $\frac{-29}{-5}{=}\frac{29}{5}$ & $\frac{-19}{-4}{=}\frac{19}{4}$ & $\frac{-11}{-2}{=}\frac{11}{2}$ &  &  & \\
$5$ & $9$ & $21$ & $\frac{-60}{-9}{=}\frac{20}{3}$ & $\frac{-42}{-8}{=}\frac{21}{4}$ & $\frac{-26}{-6}{=}\frac{13}{3}$ & $\frac{-15}{-3}{=}\mathbf{5}$ &  & \\
$6$ & $14$ & $50$ & $\frac{-104}{-14}{=}\frac{52}{7}$ & $\frac{-76}{-13}{=}\frac{76}{13}$ & $\frac{-50}{-11}{=}\frac{50}{11}$ & $\frac{-29}{-8}{=}\frac{29}{8}$ & $\frac{-17}{-4}{=}\frac{17}{4}$ & \\
$7$ & $20$ & $99$ & $\frac{-161}{-20}{=}\frac{161}{20}$ & $\frac{-121}{-19}{=}\frac{121}{19}$ & $\frac{-83}{-17}{=}\frac{83}{17}$ & $\frac{-50}{-14}{=}\frac{25}{7}$ & $\frac{-26}{-10}{=}\frac{13}{5}$ & $\frac{-16}{-5}{=}\frac{16}{5}$\\
\bottomrule
\end{tabular}
\end{table}

\begin{remark}\label{rem:z8-search-values}
The entries of Table \ref{tab:z8-search} follow at once from
\eqref{eq:theta-Lambda-Gamma}.  For $p=5$ and $n=4$, say, $\theta_5=9$ and
$\Lambda_5=2\binom53+\binom54+1-5=21$, so
$\Gamma_5(4)=21+9(8-1-10)-2\binom43-\binom44=-15$ and $\binom42-\theta_5=-3$,
giving the candidate $m=5$ and $b=2\cdot4+5-1-10=2$, which is the last line of
\eqref{eq:z8-five}. Of the five families, only the first was previously observed: the pair
$H^{1,1,t-4}$, $\bar H^{2,0,t-5}$ was identified for $5\leq t\leq11$ in
\cite{Z2Z4Z8Linearity}, where its members had to be separated with
{\sc Magma} at each individual length. The other four are new, and the last
three are invisible below $t=13$, $t=18$ and $t=22$, which is why the
computations available so far could not detect them.  
\end{remark}

The three comparisons together say precisely where a new idea is needed.
Inside the family, two infinite sequences of pairs remain unseparated by the
rank and the dimension of the kernel.  Against the classical families, one
such sequence remains for the $\Z_2\Z_4$-linear codes and five for the
$\Z_8$-linear codes.  Moreover, for the alphabets $\Z_{2^s}$ with $s\geq4$ no
comparison by the rank is possible at all, since no rank formula is known
there.  Separating these pairs, and reaching those alphabets, requires an
invariant finer than the two classical ones.

\section{Conclusions and further research}\label{sec:conclusions}

We have determined, in Theorem \ref{thm:rank-main}, the rank of the
$\Z_2\Z_4\Z_8$-linear Hadamard codes $H^{t_1,t_2,t_3}$ constructed recursively
in \cite{Z2Z4Z8Construction}, whose linearity and kernel were obtained in
\cite{Z2Z4Z8Linearity}.  The computation splits into two transparent parts.
The coordinate functions attached to the coordinates in the $\Z_2$ and the
$\Z_4$ parts span the space $V_0$ of Lemma \ref{lem:V0}, which consists of the
linear functions of the two lowest layers of message digits together with the
quadratic carries, while the top binary digit of a coordinate in the $\Z_8$
part contributes the further carry space of Proposition
\ref{prop:carry-dimension}. The technical key is Lemma \ref{lem:top-digit-reduced}: modulo $V_0$, the
top digit depends on the coordinate only through the set $S$ of positions
where the corresponding column has odd entries, and through one linear form
in the lowest digits of the message.  Every value
produced by the formula agrees with the values computed with {\sc Magma} in
\cite[Tables 2--4]{Z2Z4Z8Linearity}.

We have then used the formula, together with the dimension of the kernel, to
classify the family as far as these two invariants allow.  Inside the family,
the only pairs of distinct types of the same length sharing both invariants
are the two pairs of Theorem \ref{thm:all-collisions}.  Hence, the number of
distinct pairs $(r,k)$ is the one given in Corollary \ref{cor:rk-classes}.
For comparison with the previously studied $\Z_2\Z_4$-linear and
$\Z_8$-linear Hadamard families, the coincidences are the single family of
Theorem \ref{thm:z2z4-coincidences} and the five families of Proposition
\ref{prop:z8-coincidences}. All of these comparisons rest on the reduced rank of Definition
\ref{def:reduced-rank}, which subtracts from the rank the part determined by
the length and the kernel dimension, and thereby reduces each comparison to
a short Diophantine problem.  Along the way, we observed in Lemma
\ref{lem:Theta-all-families} that the ranks of the $\Z_2\Z_4$-linear,
$\Z_2\Z_4\Z_8$-linear, and $\Z_8$-linear Hadamard codes are values of the
same polynomial $\Theta$, evaluated at $(0,U)$, $(t_1,t_2)$, and $(a-1,b)$,
respectively.  This uniformity is also of independent interest and is what
makes the three comparisons instances of a single computation.

Two remarks on the scope of these results are in order.  First, the triple
$(t_1,t_2,t_3)$ does not determine, up to equivalence, every
$\Z_2\Z_4\Z_8$-additive Hadamard code of that abstract type.  Indeed, a
generator matrix of type $(4,6,12;2,0,1)$ whose Gray image has rank $13$,
rather than the value $12$ of $H^{2,0,1}$, is exhibited in
\cite[Example~2.4]{Z2Z4Z8Construction}.  All our statements therefore concern
the recursively constructed family, together with the constructions already
proved in \cite{Z2Z4Z8Construction} to give permutation equivalent codes.
Second, as explained in Remark \ref{rem:only-s-three}, the comparison with
the $\Z_{2^s}$-linear Hadamard codes has been carried out only for $s\leq3$,
because no rank formula is available for $s\geq4$.

Several questions remain.  The most immediate one is raised by the paper
itself.  The pairs left unseparated by Theorems \ref{thm:all-collisions} and
\ref{thm:z2z4-coincidences} and by Proposition \ref{prop:z8-coincidences} are
not pairs of equivalent codes: for $5\leq t\leq11$, they were separated in
\cite{Z2Z4Z8Linearity} by computer equivalence tests.  A complete
classification at every length therefore requires an invariant finer than
the rank and the dimension of the kernel, preferably one that can also be
evaluated without a rank formula, so as to extend the comparison to the
alphabets $\Z_{2^s}$ with $s\geq4$.

The rank of $\bar H^{a_1,\dots,a_s}$ is known only for $s\leq3$
\cite{Kro:2001:Z4_Had_Perf,PRV06,fernandez2019mathbb}, and in both cases it is
a value of the single polynomial $\Theta$ of Definition \ref{def:Theta}.  For
$s=3$, this is \eqref{eq:rank-Theta-z8}, which reads
$\rank(\bar H^{a_1,a_2,a_3})=\Theta(a_1-1,a_2)+a_1+a_3+2$.  For $s=2$ a
$\Z_4$-linear Hadamard code $\bar H^{a_1,a_2}$ has the same length and the
same dimension of the kernel as the $\Z_2\Z_4$-linear code
$H^{a_1-1,\,a_2+2}$, and is equivalent to it by \cite{KV2015}, so
\eqref{eq:rank-Theta-z2z4} gives
$\rank(\bar H^{a_1,a_2})=\Theta(0,a_1-1)+a_2+2$.  In both cases the first
parameter enters shifted by one, and the degree in $a_1$ is $2(s-1)$.  It
would be interesting to know whether the rank for $s\geq4$ is again a value of
a polynomial of the same general form.  Next, the five families of
\eqref{eq:z8-five} were obtained by solving \eqref{eq:z8-criterion}.  Although
that solution is complete, the five coincidences arise from the algebraic
form of the equation rather than from a structural explanation, and an explicit reason why
$H^{4,5,t-21}$ and $\bar H^{6,2,t-21}$ must have the same rank, would be more
satisfactory.  It would also be natural to extend the computation of Section
\ref{sec:rank} to $\Z_2\Z_4\cdots\Z_{2^s}$-linear Hadamard codes with all the
$\alpha_i$ nonzero, or to $\Z_p\Z_{p^2}\cdots\Z_{p^s}$-linear
generalised Hadamard codes with $p$ prime, in the spirit of
\cite{HadamardZps,ZpZp2Construction,ZpZp2Classification}; the method is not
tied to three alphabets, and what has to be redone in general is the reduction
of the top digit, which is where the combinatorics of the carries enters.
Finally, the case $\alpha_1=0$, $\alpha_2\not=0$, $\alpha_3\not=0$, that is,
the existence of $\Z_4\Z_8$-additive Hadamard codes, is still open, whereas
the case $\alpha_1\not=0$, $\alpha_2=0$, $\alpha_3\not=0$ cannot occur, as
observed in \cite{Z2Z4Z8Linearity}.
\bibliographystyle{elsarticle-num}
\bibliography{manuscript}

@book{Key,
  title={Designs and Their Codes},
  author={Assmus, Edvard F and Key, Jennifer D},
  @number={103},
  year={1994},
  publisher={Cambridge University Press}
}

@article{BGH83,
  title={Algebraic techniques for nonlinear codes},
  author={Bauer, Heiko and Ganter, Bernhard and Hergert, Ferdinand},
  journal={Combinatorica},
  volume={3},
  number={1},
  pages={21--33},
  year={1983},
  publisher={Springer}
}

@article{PRV06,
  title={On the additive ({$\mathbb{Z}_4$}-linear and non-{$\mathbb{Z}_4$}-linear) {H}adamard codes: rank and kernel},
  author={Phelps, Kevin T and Rif{\`a}, Joseph and Villanueva, Merc{\`e}},
  journal={IEEE Transactions on Information Theory},
  volume={52},
  number={1},
  pages={316--319},
  year={2006},
  publisher={IEEE Press Piscataway, NJ, USA}
}

@article{RSV08,
  title={On the Intersection of {$\mathbb{Z}_2\mathbb{Z}_4$}-Additive Perfect Codes},
  author={Rif{\`a}, Josep and Solov'eva, Faina Ivanovna and Villanueva, Merc{\`e}},
  journal={IEEE transactions on information theory},
  volume={54},
  number={3},
  pages={1346--1356},
  year={2008},
  publisher={IEEE}
}

@article{ccsg,
  title={{$\mathbb{Z}_2\mathbb{Z}_4$}-linear codes: generator matrices and duality},
  author={Borges, Joaquim and Fern{\'a}ndez-C{\'o}rdoba, Cristina and Pujol, Jaume and Rif{\`a}, Josep and Villanueva, Merc{\`e}},
  journal={Designs, Codes and Cryptography},
  volume={54},
  number={2},
  pages={167--179},
  year={2010},
  publisher={Springer}
}

@book{BookZ2Z4,
  author    = {Joaquim Borges and
               Cristina Fern{\'{a}}ndez{-}C{\'{o}}rdoba and
               Jaume Pujol and
               Josep Rif{\`{a}} and
               Merc{\`{e}} Villanueva},
  title     = {$\Z_2\Z_4$-Linear Codes},
  publisher = {Springer},
  year      = {2022},
  isbn      = {978-3-031-05440-2},
}

@article{Codes2k,
  title={Every {$\mathbb{Z}_{2^k}$}-code is a binary propelinear code},
  author={Borges, Joaquim and Fern{\'a}ndez-C{\'o}rdoba, Cristina and Rif{\`a}, Josep},
  journal={Electronic Notes in Discrete Mathematics},
  volume={10},
  pages={100--102},
  year={2001},
  publisher={Elsevier}
}

@article{Magma,
  title={Handbook of {M}agma functions},
  author={Bosma, Wieb and Cannon, John J and Fieker, C and Steel, A},
  journal={Edition},
  volume={2.25},
  year={2020},
  url={http://magma.maths.usyd.edu.au/magma/},
  publisher={Unknown Publisher}
}

@article{Carlet,
  title={{$\mathbb{Z}_{2^k}$}-linear codes},
  author={Carlet, Claude},
  journal={IEEE Transactions on Information Theory},
  volume={44},
  number={4},
  pages={1543--1547},
  year={1998}
}

@article{Sole,
  title={The {$\mathbb{Z}_4$}-linearity of {K}erdock, {P}reparata, {G}oethals, and related codes},
  author={Hammons, A Roger and Kumar, P Vijay and Calderbank, A Robert and Sloane, Neil JA and Sol{\'e}, Patrick},
  journal={IEEE Transactions on Information Theory},
  volume={40},
  number={2},
  pages={301--319},
  year={1994},
  publisher={IEEE}
}

@article{Kro:2001:Z4_Had_Perf,
  title={{$\mathbb{Z}_4$}-linear {H}adamard and extended perfect codes},
  author={Krotov, Denis S},
  journal={Electronic Notes in Discrete Mathematics},
  volume={6},
  pages={107--112},
  year={2001},
  publisher={Elsevier}
}

@article{KV2015,
  title={Classification of the {$\mathbb{Z}_2\mathbb{Z}_4$}-linear {H}adamard codes and their automorphism groups},
  author={Krotov, Denis S and Villanueva, Merc{\`e}},
  journal={IEEE Transactions on Information Theory},
  volume={61},
  number={2},
  pages={887--894},
  year={2015},
  publisher={IEEE}
}

@book{WMcwill,
  title={The Theory of Error-correcting Codes},
  author={MacWilliams, Florence Jessie and Sloane, Neil James Alexander},
  year={1977},
  publisher={Elsevier}
}

@article{TwoWeightSole,
  title={On two-weight {$\mathbb{Z}_{2^k}$}-codes},
  author={Shi, Minjia and Sepasdar, Zahra and Alahmadi, Adel and Sol{\'e}, Patrick},
  journal={Designs, Codes and Cryptography},
  volume={86},
  number={6},
  pages={1201--1209},
  year={2018},
  publisher={Springer}
}

@article{Blake,
  title={Codes over integer residue rings},
  author={Blake, Ian F},
  journal={Information and Control},
  volume={29},
  number={4},
  pages={295--300},
  year={1975},
  publisher={Elsevier}
}

@article{HadamardZps,
  title={On the Linearity and Classification of {$\mathbb{Z}_{p^s}$}-Linear Generalized {H}adamard Codes},
  author={Bhunia, Dipak Kumar and Fern{\'a}ndez-C{\'o}rdoba, Cristina and Villanueva, Merc{\`e}},
  journal={Designs, Codes and Cryptography},
  volume={90},
  number={},
  pages={1037--1058},
  year={2022},
  publisher={Springer}
}

@article{ZpsEquivalance,
  title={On the Equivalence of {$\mathbb{Z}_{p^s}$}-Linear Generalized {H}adamard Codes},
  author={Bhunia, Dipak Kumar and Fern{\'a}ndez-C{\'o}rdoba, Cristina and Vela, Carlos and  Villanueva, Merc{\`e}},
  journal={Designs, Codes and Cryptography},
  volume={92},
  number={},
  pages={999--1022},
  year={2024},
  publisher={Springer}
}

@article{ZpZp2Construction,
  title={On the Constructions of  {$\mathbb{Z}_p\mathbb{Z}_{p^2}$}-Linear Generalized {H}adamard Codes},
  author={Bhunia, Dipak Kumar and Fern{\'a}ndez-C{\'o}rdoba, Cristina and Villanueva, Merc{\`e}},
  journal={Finite Fields and Their Applications},
  volume={83},
  number={},
  pages={102093},
  year={2022},
  publisher={Elsevier}
}

@article{ZpZp2Classification,
  title={Linearity and Classification of {$\Z_p\Z_{p^2}$}-Linear Generalized {H}adamard Codes},
  author={Bhunia, Dipak Kumar and Fern{\'a}ndez-C{\'o}rdoba, Cristina and Villanueva, Merc{\`e}},
  journal={Finite Fields and Their Applications},
  volume={86},
  number={},
  pages={102140},
  year={2023},
  publisher={Elsevier}
}

@article{Z2Z4Z8Construction,
  title={On recursive constructions of {$\Z_2\Z_4\Z_8$}-linear {H}adamard codes},
  author={Bhunia, Dipak Kumar and Fern{\'a}ndez-C{\'o}rdoba, Cristina and Villanueva, Merc{\`e}},
  journal={Advances in Mathematics of Communications},
  volume={18},
  number={2},
  pages={455--479},
  year={2024},
  publisher={American Institute of Mathematical Sciences}
}

@article{Z2Z4Z8Linearity,
  title={Linearity and classification of {$\Z_2\Z_4\Z_8$}-linear {H}adamard codes},
  author={Bhunia, Dipak Kumar and Fern{\'a}ndez-C{\'o}rdoba, Cristina and  Villanueva, Merc{\`e}},
  journal={Designs, Codes and Cryptography},
  volume={93},
  number={},
  pages={4567--4594},
  year={2025},
  publisher={Springer}
}

@article{dougherty,
  title={Codes over {$\mathbb{Z}_{2^k}$}, {G}ray map and self-dual codes},
  author={Dougherty, Steven T and Fern{\'a}ndez-C{\'o}rdoba, Cristina},
  journal={Advances in Mathematics of Communications},
  volume={5},
  number={4},
  pages={571--588},
  year={2011},
  publisher={American Institute of Mathematical Sciences}
}

@article{EquivZ2s,
  title={Equivalences among {$\mathbb{Z}_{2^s}$}-linear {H}adamard codes},
  author={Fern{\'a}ndez-C{\'o}rdoba, Cristina and Vela, Carlos and Villanueva, Merc{\`e}},
  journal={Discrete Mathematics},
  volume={343},
  number={3},
  pages={111721},
  year={2020},
  publisher={Elsevier}
}

@article{fernandez2019mathbb,
  title={On {$\mathbb{Z}_{8}$}-linear {H}adamard codes: rank and classification},
  author={Fern{\'a}ndez-C{\'o}rdoba, Cristina and Vela, Carlos and Villanueva, Merc{\`e}},
  journal={IEEE Transactions on Information Theory},
  volume={66},
  number={2},
  pages={970--982},
  year={2019},
  publisher={IEEE}
}

@article{KernelZ2s,
  title={On {$\mathbb{Z}_{2^s}$}-linear {H}adamard codes: kernel and partial classification},
  author={Fern{\'a}ndez-C{\'o}rdoba, Cristina and Vela, Carlos and Villanueva, Merc{\`e}},
  journal={Designs, Codes and Cryptography},
  volume={87},
  number={2-3},
  pages={417--435},
  year={2019},
  publisher={Springer}
}

@article{Krotov:2007,
  title={On  {$\mathbb{Z}_{2^k}$}-dual binary codes},
  author={Krotov, Denis S},
  journal={IEEE Transactions on Information Theory},
  volume={53},
  number={4},
  pages={1532--1537},
  year={2007},
  publisher={IEEE}
}

@article{Nechaev,
  title={Weighted modules and representations of codes},
  author={Honold, Th. and Nechaev, Aleksandr Aleksandrovich},
  journal={Probl. Inf. Transm.},
  volume={35},
  number={3},
  pages={205--223},
  year={1999},
  publisher={}
}

@article{Shankar,
  title={On {BCH} codes over arbitrary integer rings},
  author={Shankar, Priti},
  journal={IEEE Transactions on Information Theory},
  volume={25},
  number={4},
  pages={480--483},
  year={1979},
  publisher={IEEE}
}

@article{ShiKrotov2019,
  title={On  {$\mathbb{Z}_p\mathbb{Z}_{p^k}$}-additive codes and their duality},
  author={Shi, Minjia and Wu, Rongsheng and Krotov, Denis S},
  journal={IEEE Transactions on Information Theory},
  volume={65},
  number={6},
  pages={3841--3847},
  year={2019},
  publisher={IEEE}
}

@article{ShiTwoHomWeight,
  author    = {Minjia Shi and
               Thomas Honold and
               Patrick Sol{\'{e}} and
               Yunzhen Qiu and
               Rongsheng Wu and
               Zahra Sepasdar},
  title     = {The geometry of two-weight codes over {$\mathbb{Z}_{p^m}$}},
  journal   = {{IEEE} Transactions on Information Theory},
  volume    = {67},
  number    = {12},
  pages     = {7769--7781},
  year      = {2021}
}

@article{H07,
  author    = {Horadam, K. J.},
  title     = {{H}adamard matrices and their applications},
  journal   = {Princeton University Press},
  volume    = {},
  number    = {},
  pages     = {},
  year      = {2007}
}

@article{SBT98,
  author    = {Smith, E. D. J. and Blaikie, R. J. and Taylor D. P.},
  title     = {Performance enhancement of spectral-amplitude coding optical CDMA using pulse-position modulation},
  journal   = {{IEEE} Transactions on Communications},
  volume    = {46},
  number    = {},
  pages     = {1176--1185},
  year      = {1998}
}

@article{HYT04,
  author    = {Huang, J. F. and Yang, C. C. and Tseng, S. P.},
  title     = {Complementary {W}alsh-{H}adamard coded optical {CDMA} coder/decoders structured over arrayed-waveguide grating routers},
  journal   = {Opt. Commun.},
  volume    = {229},
  number    = {},
  pages     = {241--248},
  year      = {2004}
}

@article{Nyb91,
  author    = {Nyberg, K.},
  title     = {Perfect nonlinear {S}-boxes},
  journal   = {{EUROCRYPT-91}},
  volume    = {LNCS 547},
  number    = {},
  pages     = {378--385},
  year      = {1991},
 publisher={Springer}
}

@article{YLL03,
  author    = {Yu, G. J. and Lu, C. S. and Liao, H. Y.},
  title     = {A message-based cocktail watermarking system},
  journal   = {Pattern Recognition},
  volume    = {36},
  number    = {},
  pages     = {957--968},
  year      = {2003}
}
\end{document}